\documentclass[11pt,letterpaper]{article}

\usepackage[margin=1in]{geometry}
\usepackage[english]{babel}
\usepackage{amsmath}
\usepackage{amsfonts}
\usepackage{amssymb}
\usepackage{amsthm}
\usepackage[ruled,vlined]{algorithm2e}
\usepackage{graphicx}
\usepackage{todonotes}
\usepackage{array,multirow}
\usepackage{authblk}
\usepackage{hyperref}

\usepackage{thmtools, thm-restate}
\declaretheorem{theorem}

\newtheorem{lemma}{Lemma}

\usepackage[style=ieee]{biblatex}
\newcommand{\diam}{\mathrm{diam}}

\title{Fault-Tolerant Distributed Directories}
\author[1]{Judith Beesterm\"oller \thanks{jbeesterm@ethz.ch}}
\author[2]{Costas Busch \thanks{kbusch@augusta.edu}}
\author[1]{Roger Wattenhofer \thanks{wattenhofer@ethz.ch}}
\affil[1]{ETH Zurich}
\affil[2]{Augusta University}

\date{}

\begin{document}
\maketitle

\setcounter{page}{0}
\thispagestyle{empty}

\begin{abstract}
    A distributed directory is an overlay data structure on a graph $G$ that helps to access a shared token $t$. The directory supports three operations: {\em publish}, to announce the token, {\em lookup}, to read the contents of the token, and {\em move}, to get exclusive update access to the token. The directory is built upon a hierarchical partition of the graph using either weak or strong clusters. The main mechanism is the maintenance of a {\em directory path} that starts at the root node in the hierarchy and points to the current owner of the token. In the literature, there are known directory algorithms based on hierarchical graph structures, but none of them have considered failures. Given a hierarchical partition, we consider the impact of $f$ edge failures on the functionality and performance of the distributed directory. The edge failures may result in the splitting of clusters into up to $f+1$ connected components and an increase in the number of levels in the hierarchy. To recover the hierarchical partition after failures, we maintain spanning trees in the clusters and their connected components. If $G$ remains connected, we show that each level of the directory path is dilated by only a factor $f$. We also show that the performance of the publish and lookup operations is affected in the worst case by a factor $f$ with respect to the message complexity. The message complexity of the move operation increases by an $f$ factor and the number of additional layers in the hierarchy.
\end{abstract}
\newpage

\section{Introduction}
    We study distributed directories that facilitate access to a shared token $t$ on a weighted graph $G$. The token resides in one of the nodes in the graph which is its current owner. The directory helps to find the token and gives access to it. It supports three operations: (i) {\em publish}, which is used to announce the ownership of the token, (ii) {\em lookup}, which is used to read the contents of the token, and (iii) {\em move}, which is used to move the token to a new owner to obtain exclusive access. These operations may be issued concurrently by the nodes in $G$. Distributed directories that support such operations have been used for distributed transactional memory \cite{Herlihy2007Distributed,Sharma2014distributed}, distributed queues \cite{Tirthapura2006Self, Herlihy2001Competitive}, and mobile object tracking \cite{Baruch1991Concurrent,Sharma2015Near}.
    
    To perform the directory operations, the nodes communicate with messages. The weight $w(e)$ of an edge $e$ represents the cost of sending a message over it. The cost of sending a message from node $u$ to node $v$ is the length of a shortest path between them in $G$ (assuming fixed-size messages). We are interested in solutions with a small message cost per operation. For a lookup operation issued by a node $u$, the minimum message cost is the distance in $G$ between $u$ and the current owner node of token $t$. For a set of $k$ sequential move operations issued by a sequence of nodes $v_1, v_2, \ldots, v_k$, the minimum communication cost to serve all requests is the sum of the distances between consecutive pairs $v_{i}, v_{i+1}$, $1 \leq i < k$.
    
    There are known efficient directory schemes that achieve message complexity close to the optimal (within poly-log factors) \cite{Herlihy2007Distributed,Rai2022load,Sharma2014distributed,Sharma2015Analysis}. These are based on an appropriate hierarchical cluster decomposition of $G$. However, these approaches do not consider failures. In reality, a directory is implemented on a network modeled as a graph $G$. It is common to have unreliable networks with link (edge) failures between nodes. Edge failures result in changing the distances in $G$, and will directly affect the performance of the directory. We study the impact of $f$ edge failures. We examine the costs of repairing the directory and the impact on the operations. We first present the basic directory approach. Then we discuss our contributions with respect to edge failures.
    
        Our directory approach is inspired by the Spiral protocol \cite{Sharma2014distributed} which uses a sparse cover decomposition (clusters may overlap); instead, we use a sparse partition hierarchy $\mathcal{P}$ to study the impact of failures. 
        
        The hierarchy ${\mathcal P}$ consists of $\mathcal{O}(\log D)$ levels, where $D$ is the diameter of $G$. The partition levels are obtained from a $(\sigma, I)$-sparse partition scheme of $G$. For an appropriate $\rho > 1$ (typically a constant), level $i$ is a partition of $G$ into clusters of nodes such that each node belongs to exactly one cluster, each cluster has diameter at most $\sigma \rho^i$, and the $\rho^i$-neighborhood of a node spans at most $I$ clusters.
        
        The clusters in ${\mathcal P}$ can have a weak diameter, as measured with respect to the whole of $G$, or a strong diameter as measured with respect to the cluster's induced subgraph. There are known $(\mathcal{O}(\log n), \mathcal{O}(\log n))$-sparse partition schemes for arbitrary graphs with $n$ nodes, for both weak diameter and strong diameter \cite{jia2005universal,filtser2020scattering}. These improve to $(\mathcal{O}(1), \mathcal{O}(1))$-sparse partition schemes for weak diameter in constant doubling-dimension graphs and fixed minor-free (e.g. planar) graphs.
        
        The distributed directory is implemented by maintaining a directory path $\phi$ from a root node in $G$ to the current owner of token $t$. For this purpose we pick a leader in each cluster of ${\mathcal P}$. At the highest level, there is a single cluster with the root as leader, while at the lowest level each cluster is an individual node of $G$. The directory path is a sequence of pointers from the root toward the owner node of $t$, going through all the intermediate levels of ${\mathcal P}$. 
        
        The directory path $\phi$ is created by a publish operation issued by the first owner of $t$. There will only be one publish operation for $t$, after that $\phi$ is updated at every move operation. A lookup operation issued by node $v$ searches for $\phi$ ``upward'' in the hierarchy by checking all the leaders in the clusters in the $\rho^i$-neighborhood of $v$, for all increasing levels $i$. When it finds a leader that knows about $\phi$, the lookup operation follows the directory path toward the token. A move operation builds a new directory path toward the new owner and it deletes the old one while it searches for $\phi$. The token is then transferred to the new owner. The modifications of the directory path by a move operation make it harder for future lookup operations to discover it. Nevertheless, we will show that the costs of a lookup or move operation are always close to optimal.
        
        We obtain the following bounds for the message complexity for general graphs (see Table \ref{table:performance}). The length of the initial directory path is $\mathcal{O}(\sigma D) = \mathcal{O}(D \cdot \log n)$, which is also the cost of the publish operation. For lookup, the cost is an $\mathcal{O}(\sigma^2 \rho I) = \mathcal{O}(\log^3 n)$ approximation of the optimal. For move, the amortized cost of a sequence of move operations is an $\mathcal{O}( h \rho \sigma (\sigma + I)) = \mathcal{O}( \log D \cdot \log^2 n)$ approximation of the optimal, where $h=\lceil \log_\rho(D)\rceil$ is the number of levels in $\mathcal{P}$. For weak diameter and special kinds of graphs that use a $(\mathcal{O}(1),\mathcal{O}(1))$-partition scheme, we get better bounds which are $\mathcal{O}(1)$ approximation for lookup, and $\mathcal{O}(\log D)$ approximation for move, while the publish cost is simply $O(D)$.
        
    \subsection{Edge Failures}
        We consider the impact of $f \geq 1$ edge failures on the directory structure. An edge failure results in the removal of an edge from $G$. We make the assumption that the edge failures do not disconnect $G$, since otherwise, the owner of the token $t$ may become unreachable by other nodes in $G$ which makes the directory unusable. 
        
        The $f$ edge failures may happen at arbitrary moments and can occur concurrently. The goal is to maintain and update ${\mathcal P}$ and the directory path whenever failures occur. We provide mechanisms to handle the failures and respective repairs dynamically, on the fly as failures happen, and in a distributed manner without disrupting ongoing directory operation requests.
        
        Edge failures have adverse effects on the diameter and connectivity of the clusters. In an unweighted graph, $G$ the diameter of $G$ can at most double with every edge failure \cite{chung1984diameter}. However, in an unweighted graph, we cannot bound the diameter, as seen in Figure~\ref{fig: unbounded_diam}. To ensure our directory continues to work efficiently, we add additional layers to our hierarchy.

        \begin{figure}
            \centering
            \includegraphics{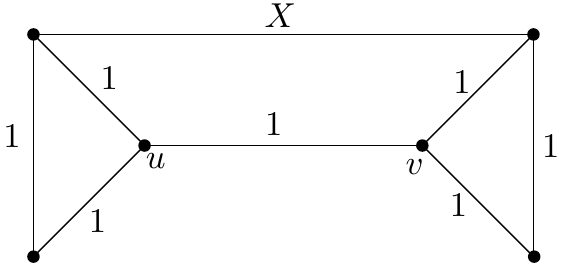}
            \caption{An example of a graph showing that we cannot bound the stretch in the diameter of the graph. Assuming $X>1$ the initial graph has diameter 3. If edge $\{u, v\}$ fails, the diameter becomes $2+X$. Without further assumptions on $X$ this cannot be bounded as a constant multiple of $3$.}
            \label{fig: unbounded_diam}
        \end{figure}

        \begin{table}[t]
                \centering
                \begin{small}
                \begin{tabular}{| l | l | l || l | l | l |}
                \hline\hline
                Graph  & Partition & Failures & Publish & Lookup & Move \\
                \hline\hline
                general & any & none & $\mathcal{O}(D \cdot \log n)$ & $\mathcal{O}(\log^3 n)$ & $\mathcal{O}( \log D \cdot \log^2 n)$\\
                \hline
                general & weak & $f$ & $\mathcal{O}(D' \cdot \log n)$ & $\mathcal{O}(f \cdot \log^3 n)$ &  $\mathcal{O}( f\log D'\log^2(n))$ \\
                \hline
                general & strong & $f$ &  $\mathcal{O}(D' \cdot \log n)$ & $\mathcal{O}(f \cdot \log^2 n + \log^3 n)$ & $\mathcal{O}( \log D'\log n(f+\log n))$ \\
                \hline
                special & any & none & $\mathcal{O}(D)$ & $\mathcal{O}(1)$ & $\mathcal{O}(\log D)$\\
                \hline
                special & weak & $f$ & $\mathcal{O}(D')$ & $\mathcal{O}(f)$ & $\mathcal{O}(f \cdot \log D')$\\
                \hline\hline
                \end{tabular}
                \end{small}
                \caption{Cost of operations for general/special graphs and weak/strong diameter partitions; publish cost is absolute; lookup and move costs are approximation factors; failures are $f \geq 1$; $D'$ is the diameter of $G$ after the $f$ failures; special graphs are constant doubling dimension graphs and fixed minor-free graphs with $\sigma, I \in \mathcal{O}(1)$}
                \label{table:performance}
            \end{table}

        \paragraph{Cluster restructuring after failures.}
            To mitigate the effects of edge failures on a cluster $X$ at level $i<h$, we split $X$ into up to $f+1$ new clusters such that the diameter of each new cluster is no more than twice the original diameter of $X$. To update a single cluster and inform all affected nodes requires $\mathcal{O}(n^2)$ messages that each have size and cost $\mathcal{O}(\log(n))$. The costs of the other repair mechanisms are displayed in Table \ref{table: repair}. To mitigate the effect of the increase in the overall graph diameter, we increase the number of layers in the partition to $\log_\rho(D')$, where $D'$ denotes the diameter of the graph after the edge failures. 

            A cluster split affects the parameters $\sigma$ and $I$ of the partition hierarchy ${\mathcal P}$. Since the diameter of the new cluster doubles, $\sigma$ changes to $2 \sigma$. The parameter $I$ is also affected, to give $(2\sigma, (f+1)I)$-partition scheme for weak diameter and $(2\sigma, f+I)$-partition scheme for strong diameter. The results on the performance of the three operations scale accordingly with $f$ and $D'$ as shown in Table \ref{table:performance}. Strong partitions respond asymptotically better to failures in general graphs when $f = \Omega(\log n)$. However, weak diameter partitions are better for graphs having $\sigma, I \in \mathcal{O}(1)$.
            
            The technique that we use to achieve constant dilation in the cluster diameter is by maintaining a spanning tree of the clusters. Each time an edge $e$ of $X$ fails, we split $X$ into two new clusters where the first has the original leader of $X$, while the second has as leader one of the incident nodes of $e$. In weak diameter clusters, the edge $e$ may actually reside outside $X$, in which case the new leader is chosen appropriately in $X$.
            
            In weak diameter partitions, an edge removal may affect multiple clusters at level $i$, since an edge may be in the spanning tree of multiple clusters at level $i$. Hence, $I$ increases to $(f+1)I$, since each of the (at most) $I$ clusters in a $\rho^i$-neighborhood at level $i$ is split to at most $f+1$ clusters. In strong diameter partitions, each edge removal affects at most one cluster at level $i$. Hence, the total number of new clusters at level $i$ increases by $f$ and $I$ becomes $f+I$. 
            
        \paragraph{Maintaining the directory path.} 
            If a cluster $X$ whose leader node is part of the directory path is split due to the edge failure, we need to determine if the directory path needs to be updated as well. The cluster leader $l(X)$ is added to the directory path by a node $w$ in $X$ during a publish or move operation. After the splitting of $X$, we update the directory path such that it contains the leader node of the cluster that contains $w$.
            
            A leader node $l(X)$ that is added to the directory path through a publish or move operation executed by node $w$ remembers $w$. When the splitting of cluster $X$ occurs, the leader node $l(X)$ determines if $w$ is still connected to it on the spanning tree of $X$. If $w$ got disconnected, then $l(X)$ sends a message to $l(X_2)$ to inform it that it is part of the directory path. The message includes information about the leader nodes on the directory path on the level above and below so that $l(X_2)$ knows how to set its pointers to connect to the directory path. Furthermore, the message tells $l(X_2)$ the id of the node that added $l(X)$ to the directory path, so that in case of another failure $l(X_2)$ can determine if another directory path update is needed. When a cluster $X$ splits and the partitioned off cluster is not part of the directory path, then the leader $l(X)$ will nonetheless send a message to $l(X_2)$ informing it that it is not part of the directory path. 
            
            The mechanism that updates the directory path can be executed during the execution of the directory protocol. If $w$ was the last node to modify the directory path at level $i$, then any node that contacts the cluster leader of the cluster containing $w$ will find the directory path at level $i$.

            If we add additional layers to the partition hierarchy, we also need to extend the directory path to these layers. If the root node was last informed about a change in the directory path by a node $w$, then $r$ insures that in the additional layers the cluster that contains node $w$ is part of the directory path. 
        
            \begin{table}[t]
            \resizebox{\textwidth}{!}{%
                \centering
                \begin{small}
                \begin{tabular}{ |p{2.7cm}||p{4.2cm}|p{2.5cm}|p{3.2cm}|p{2.5cm}|  }
                \hline\hline
                Operation & Partition & Size of Message & Number of Messages & Cost of Message\\
                \hline\hline
                \multirow{2}{3cm}{Reclustering (per cluster at level $i$)} & Any & $\mathcal{O}(\log(n))$ & $\mathcal{O}(n^2)$ & $\mathcal{O}(\log(n))$\\\cline{2-5}
                & weak (additional to above) & $\mathcal{O}(n\log(n))$ & 1 & $\mathcal{O}(\log(n))$\\\hline
                Updating Shortest Path Trees & Any & $\mathcal{O}(\log(n))$ & $\mathcal{O}(n^2)$ & $\mathcal{O}(D')$\\\hline
                Update Directory Path at Level $i$ & Any & $\mathcal{O}(\log(n))$ & $\mathcal{O}(1)$ & $\mathcal{O}(D')$\\
                \hline\hline
                \end{tabular}
                \end{small}
                }
                \caption{Cost of repair mechanism for general graphs and strong/weak partitions; $D'$ denotes the diameter of the underlying graph after the edge failure}
                \label{table: repair}
            \end{table}

    \subsection{Related Work}
        
        An alternative way to implement a distributed directory is to use a spanning tree $T$ on $G$. The edges of $T$ are directed towards the owner node of the token. If a node $u$ requests to obtain the token, then the move request redirects the edges of the tree toward $u$ (edge reversal). The benefit of the tree is that it can easily handle distributed requests since concurrent move operations are ordered when they intersect on the tree. Several protocols have been proposed based on trees (Arrow \cite{Kerry1989TreeBased,Demmer1998TheArrow,Herlihy2006Dynamic,Kuhn2004Dynamic}, Relay \cite{Zhang2010Dynamic}, Ivy \cite{Kai1989Memory}, Arvy \cite{Khanchandari2020Arvy}). The approximation factor of the operations is $O(\log D_T)$, with respect to the diameter $D_T$ of tree $T$. However, by using a tree $T$ the performance of the lookup and move operations may be sub-optimal with respect to $G$, as $T$ may not accurately represent the shortest path distances on $G$. Considering the distance stretch $s$ of the tree the approximation becomes $O(s \cdot \log D_T)$, and $s$ can be as large as the graph diameter $D$. Nevertheless, considering an appropriate overlay tree that preserves on average the pairwise node distances of $G$ \cite{Fakcharoenphol2004Tight}, it is possible to get close to optimal performance on the average case for a set of random source operation requests \cite{Peleg1999Variant,Ghodselahi2017Dynamic}. Our approach, on the other hand, has guaranteed performance for arbitrary sources of requests (not just random).
        
        \paragraph{Outline of the Paper.}
        In Section \ref{sec:model}, we give the model of the paper (with details on sparse partitions in Section \ref{app: sparse partition}). In Section \ref{sec: Directory Scheme}, we present the basic directory scheme without failures (the pseudocode is in Section \ref{sec: pesudocode}). Section \ref{sec: Responding to Edge Failures} describes the re-clustering and directory path updates due to edge failures (concurrent edge failures are discussed in Section \ref{app: concurrent edge failures}, and handling dynamically the ongoing operations in Section \ref{app: Dynamic Setting}). The message complexity is given in Section \ref{sec:analysis} (full proofs are deferred to Sections \ref{app: transient operations}, \ref{sec: proofs}, \ref{sec: Analysis without faults}, and \ref{sec: cost fault mechanism}). We conclude in Section \ref{sec: conclusion}.
         

\section{Model} \label{sec:model}
    We model the distributed system as a weighted graph $G=(V, E, w)$ with positive edge weights of at least 1. Nodes communicate with each other asynchronously through messages, but messages sent along the same edge are delivered and processed in the order they are sent. All messages have the same size (typically logarithmic in $n$). The cost of sending a message over an edge is the weight of the edge. While handling edge failures to update the clusters and the directory path, we may use larger size messages. The involved message complexities are explicitly stated in our results. 
    
    Each node $u$ stores a shortest path tree, denoted $T(u)$, with $root(T(u))=u$. Unless stated otherwise, messages are sent along shortest paths. A node knows for each of its incident edges which shortest path trees the edge belongs to, and remembers all messages sent an received on it.
    
    The directory is built on top of a sparse partition hierarchy ${\mathcal P}$. Please refer to Section \ref{app: sparse partition} for a review of sparse partitions. Our mechanisms work for both strong and weak sparse partitions, though, as we will show, they differ in their performance. We write $\diam(X)$ to refer to the strong, respective weak diameter of cluster $X$ depending on the type of partition used.
    
    
    In every cluster $X$ of ${\cal P}$, we select a {\em leader} $l(X)$. At $P_{-1}$ each node is a leader of its own cluster. The leader of the single cluster of $P_h$ is called {\em root}. Every node $u \in V$ belongs to exactly one cluster in each level of ${\cal P}$. Denote by $C_i(u)$ the cluster of $u$ at level $i$, and let $l_i(u) = l(C_i(u))$ be its leader.

    For each cluster $X$ we store a spanning tree $T(X)$. Initially, $T(X)$ is a shortest path tree with $root(T(X)) = l(X)$. If our partition hierarchy $\mathcal{P}$ is based on a sparse partitions, then $T(X)$ might contain some nodes that are not in $X$, otherwise (strong partitions), $T(X)$ contains only nodes from $X$. The choice of spanning tree guarantees that for any node $u$ in $X$ the path on $T(X)$ connecting $u$ and $l(X)$ is at most $\diam(X)$. For any node $u$ on $T(X)$, we denote the subtree of $T(X)$ rooted in $u$ by $T_{\setminus u}(X)$. Every node $u$ on $T(X)$ knows $T_{\setminus u}(X)$ and $l(X)$. If our data structure is built on weak sparse partitions, then $u$ also knows which nodes in $T_{\setminus u}(X)$ belong to $X$.
    
    Before we start the protocol, every node $u$ determines which vertices are at distance at most $r_i$ using its shortest path tree, for all $0\leq i \leq h$. Node $u$ then queries these nodes for the leader node of their cluster at level $i$, and stores this information.


\section{Directory Scheme}\label{sec: Directory Scheme}
    The directory supports three operations: {\em publish} that announces the generation of the token, {\em lookup} that helps find the token, and {\em move} that requests the token to move to a new owner node.
    
    Our hierarchy $\mathcal{P}$ consists of $h+1$ levels, where $h=\mathcal{O}(\log(n))$. At level $-1$ the clusters consist of individual nodes. The diameter of a cluster at level $0\leq i\leq h$ is at most $\sigma r_i$, where $r_i = \min\{D, \rho^i\}$.
    
    The token resides at an owner node at the lowest level. There is a virtual {\em directory path} $\phi$ that points to the current position of the token. The path $\phi$ consists of $h+2$ leader nodes, one node at every level of ${\cal P}$. Let $\phi_i$ denote the leader node of $\phi$ at level $i$, for $-1 \leq i \leq h$, where $\phi_{-1}$ is the token owner and $\phi_h$ is the root. For each level $i$, $0 \leq i \leq h-1$, leader $\phi_i$ has pointers to the lower level leader $\phi_{i-1}$ and the upper level leader $\phi_{i+1}$ which form a virtual doubly linked list.
    
    The directory path is created in publish operation. The creator node $u$ (original owner) sends a message to its leader nodes at every level of ${\cal P}$, namely $\phi_i = l_i(u)$. This operations also creates the pointers from $\phi_i$ to $\phi_{i-1}$ (both ways) for $0 \leq i \leq h$.
    
    The following theorem helps to bound the cost of the publish operation.
    \begin{theorem}
        The initial directory path has length $\mathcal{O}(\sigma\rho^h)$.
    \end{theorem}
    
    \begin{proof}
        The distance between consecutive nodes on the directory path is at most $d(l_i(u), l_{i+1}(u))\leq d(l_i(u), u) + d(u, l_{i+1}(u))$. A cluster at level $i$ has a diameter of at most $r_i\sigma$. Thus,
        \begin{align*}
            length(\phi) 
            \leq\sum_{i=-1}^{h-1} \sigma(r_i + r_{i+1})
            \leq \sum_{i=-1}^{h-1} \sigma(\rho^i + \rho^{i+1})
            \leq \frac{\sigma(\rho + 1)(\rho^{h+1}-1)}{(\rho -1) \rho}. & \qedhere
        \end{align*}
    \end{proof}
    
    To perform lookup and move operations, a node has to first discover the directory path. Once a leader node of the directory path is discovered, the current token owner can be reached by following the directory path toward the object. A lookup operation leaves the directory path intact, whereas a move operation modifies the directory path to point to the new owner of the token.
    
    A lookup operation first discovers a leader node of the directory path, say $\phi_j$ at level $j$, and then reaches the token by following the directory path down to the owner node. It then reads the contents of the token and returns the result to the issuer. The discovery of the directory path happens by exploring leader nodes close to $v$ at increasing levels of ${\mathcal P}$. Let $P_i(v)$ be the set of clusters of partition $P_i$ that intersect with $N_{G,v}(r_i)$; note that $|P_i(v)| \leq I$. For each level $i$, starting at $i=0$ up to $i = h$, node $v$ checks whether any of the leaders in $P_i(v)$ equals $\phi_i$. The exploration stops at the lowest level that a directory path leader is found (at the root in the worst case).

    \begin{lemma}\label{lem: lookup path}
        The path traversed by a lookup operation at level $i$ has length at most $\mathcal{O}(r_i\sigma I)$. 
    \end{lemma}
    
    \begin{proof}
        Each leader in $P_i(u)$ needs to be contacted only once, even if several nodes in the cluster are at a distance at most $r_i$ from $u$. Let $X$ be a cluster in $P_i(u)$. Then there exists a node $x\in X$ with $d(u, x)\leq r_i$. This implies that $d(u, l(X)) \leq d(u, x) + d(x, l(X))\leq r_i + \sigma r_i$. By construction, $|P_i(u)|\leq I$. Therefore, the distance to contact all leader nodes in $P_i(u)$ is at most $I(1+\sigma)r_i$. \qedhere
    \end{proof}
    
    A move operation issued by node $v$ will modify the directory path to point to $v$, denoting the new ownership at $v$. Similar to the lookup, a move operation first discovers a leader node of the directory path by exploring the leaders of $P_i(v)$ at each level $i$, starting at $i = 0$ and going upwards until possibly $i = h$. Let $\phi_j$, $j \geq 0$, be the first leader of $\phi$ that the move operation discovers. The formation of the new directory path is done while searching for the directory path. Node $v$ first adds $l_{-1}(v)$ to the directory path. For levels $0\leq i <j$ node $v$ first searches its $P_i(v)$ neighborhood, 
    and adds $l_i(v)$ to the director path, when it does not find $\phi$. At level $j$ node $v$ will contact leader $\phi_j$, which will remain then in the directory path but a new pointer will be established pointing from $\phi_j$ to $\phi'_{j-1}$. Hence, the new directory path is $\phi_h\cdots\phi_j \phi'_{j-1} \cdots \phi'_{-1}$.

    Once the move operation reaches $\phi_j$, it follows downward the old directory path toward the old owner $\phi_{-1}$ which it informs that the new token owner is $v$. While going downward the move operation deletes the leaders from the old directory path, that is, $\phi_{j-1} \cdots \phi_{-1}$ are removed from $\phi$.
    
    Multiple nodes may issue move operations concurrently, resulting in one or more newly formed sub-paths which will merge with the directory path. However, only the latest version of the directory path will include the root node. To guarantee that at any moment there is a unique complete directory path (without splits or gaps) from the root to an owner node, we make sure that in the upward phase of the move operation the exploration and the updates happen atomically. When a move operation contacts $\phi_i$ to query it about the directory path, then the message will cause $\phi_i$ to immediately update its downward pointer to $\phi'_{i-1}$. Subsequent lookup and move operations that reach $\phi_i$ will thus be directed to the new directory path. As the directory path $\phi$ develops from sub-paths that were initiated by different source nodes, the length of the directory path can increase compared to the initial directory path.
    
    \begin{lemma}\label{lem: distance consecutive nodes}
        The distance between two consecutive nodes $\phi_i$ and $\phi_{i+1}$ on the directory path is at most $d(\phi_i, \phi_{i+1}) \leq \sigma(r_i +r_{i+1}) + r_{i+1}$, for $-1\leq i < h$.
    \end{lemma}
    
    \begin{proof}
        Nodes $\phi_i$ and $\phi_{i+1}$ were either added to the directory path by the same node $v$, in a publish or move operation, or $v$ added $\phi_i$ and then found $\phi_{i+1}$ during its search for the directory path at level $i+1$. In the first case, $\phi_i=l_i(v)$ and $\phi_{i+1}=l_{i+1}(v)$. Therefore, their distance is bounded by 
            $$d(\phi_i, \phi_{i+1})\leq d(\phi_i, v) + d(v, \phi_{i+1}) \leq \sigma (r_i + r_{i+1}).$$ 
        In the second case, node $v$ found $\phi_{i+1}$, because it is the leader of a cluster that contains a node $w$, such that $d(v, w)\leq r_{i+1}$. Therefore, in this case we have
            \begin{equation*}
                d(\phi_i, \phi_{i+1})\leq d(\phi_i, v) + d(v, w) + d(w, \phi_{i+1})\leq \sigma(r_i +r_{i+1}) + r_{i+1}. \qedhere
            \end{equation*}
    \end{proof}
    The total length of the directory path is simply the sum of distances between consecutive nodes. 
    
    \begin{lemma}\label{lem: Length directory path}
        The directory path from level $-1$ to level $i$ has length $\mathcal{O}(\sigma \rho^{i})$.
    \end{lemma} 
    
    \begin{proof}
        In theory, case two above could apply for every pair of consecutive nodes on $\phi$. Therefore,
        \begin{align*}
            \text{length }(\phi_{-1},\cdots, \phi_i) \leq \sum_{j=-1}^{i-1}d(\phi_j, \phi_{j+1})
            \leq \sum_{j=-1}^{i-1} \sigma(r_j +r_{j+1}) + r_{j+1}
            \leq \frac{(\rho^{i+1} -1)(\sigma\rho + \rho + \sigma)}{(\rho -1)\rho}.&\qedhere
        \end{align*}
    \end{proof}
    
    The directory path to node $w$ may not be comprised completely by leaders in clusters that contain $w$. Hence, a lookup operation issued by $u$ might not discover the directory path at the lowest level where $u$ contacts the leader of $w$. To alleviate this problem, a leader node $l_i(w)$ that is added to $\phi$ informs leader node $l_{i'}(w)$, where $i' = i +  \log_{\rho} (c' \sigma)$, for an appropriately constant $c'$. We call $l_{i'}(w)$ the special parent of $w$ with respect to level $i$. When node $u$ searches for the directory path, it asks the leader nodes if they are part of the directory path, or if they are the special parent of a node on the directory path. If it finds a node that is a special parent of a node on the directory path, it takes the link to the node on the directory path and continues to follow the directory path downwards from there. In the downward phase of a move the information of the respective special parents about the previous source is also deleted. It can happen that a move removes a node from the directory path at level $i$, while a lookup follows a link from a special parent to that node. In this case, the lookup will take the link back to level $i'$ and continues its search there. 
    

\section{Responding to Edge Failures}\label{sec: Responding to Edge Failures}
    In case of edge failures, our clustering is not guaranteed to satisfy the properties of a sparse partition. Additionally, some of the shortest trees that the nodes store could become disconnected. Hence we need to repair our data structures, to guarantee correctness and performance of our algorithm. 
    
    Our repair mechanism consists of several parts: If the failed edge $e$ lies on the spanning tree of cluster $X$, we split $X$ into two clusters. Second, we update all the shortest path trees that contain $e$. If needed, we then update the directory path. Finally, we ensure that no message is lost. 
    
    To recluster a single cluster at level $i$, we need $\mathcal{O}(n^2)$ messages, each of size $\mathcal{O}(\log(n))$ and cost $\sigma r_i$. In a weak sparse partition, one additional message of size $\mathcal{O}(n\log(n))$ and cost $\sigma r_i$ is needed. To update the directory path at level $i$, we need $\mathcal{O}(1)$ additional messages of size $\mathcal{O}(\log(n))$, that have a cost of up to $D'$, where $D'$ denotes the diameter of the graph remaining after the edge failure. Updating all shortest path trees requires up to $\mathcal{O}(n^2)$ messages of size $\mathcal{O}(\log(n))$ and cost $\mathcal{O}(D')$. A detailed cost analysis can be found in Section \ref{sec: cost fault mechanism}.
    
    We modify the operations described in Section \ref{sec: Directory Scheme} as follows:
    \begin{itemize}\setlength\itemsep{0em}
        \item Each node on $\phi$ knows which node $w$ added it to $\phi$ during a publish or move initiated by $w$.
        \item When $w$ contacts a leader node $l(X)$ during the search for the directory path in a lookup, or move operation, then the message includes a list of all nodes from $w$'s $r_i$-neighborhood that $w$ believes to be part of cluster $X$.
        \item Node $w$ contacts a level $i$ leader node $l(X)$ only if  $d(w, l(X))\leq r_i + 2\sigma r_i$.
    \end{itemize}

    \subsection{Updating Shortest Path Trees}\label{subsec: upadating shortest path trees}
        In our protocol, every node stores a shortest path tree. When an edge $e=\{u, v\}$ fails, we need to update all the ones that contain edge $e$. To update a shortest path tree we use the fully dynamic algorithm for maintaining a shortest path tree in the presence of edge failure developed by King~\cite{king1999fully}. King's algorithm mimics Dijkstra to update the tree. To update all shortest path trees affected by the edge failure, we simply call the algorithm for every tree that needs to update. 

        In the process of updating its shortest path tree, node $w$ determines for which nodes the distance from $w$ increased, and updates the information regarding which nodes to contact at which level accordingly. Once $w$ updated its shortest path tree, it informs all nodes about the changes in its shortest path tree. That is, it sends a message to the endpoints of those edges that were added or removed from the shortest path tree.

        To initialize the update of the shortest path tree, we recall that every node knows, for each of its incident edges, which shortest path trees the edge belongs to. When edge $e=\{u, v\}$ fails, both $u$ and $v$ detect the failures. For each shortest path tree that contains $e$ either the path to $u$ or to $v$ is not affected by the edge failure. Hence, this node can inform the root node to initialize an update. 

    \subsection{Reclustering}
        Suppose that edge $e=\{u, v\}$ fails. For the reclustering, we distinguish between the root level and clusters below the root level. Consider first, the clusters that are not at the root level. For a cluster $X$, whose spanning tree $T_X$ contains edge $e$, we split $X$ into two clusters defined by the connected components of $T_X\setminus\{e\}$, i.e., $X_1= X\setminus (X\cap T_{\setminus v}(X))$ and $X_2 = X\cap T_{\setminus v}(X)$. Here, we assume $d(u, l(X))< d(v, l(X))$ on $T(X)$. We define $l(X_1) = l(X)$. If we use a strong sparse partition, then $l(X_2) = v$. Otherwise, $v$ chooses a node $w$ closest to it on $T(X)$ from $X\cap T_{\setminus v}(X)$ to become $l(X_2)$. The tree $T(X)\setminus T_{\setminus v}(X)$ forms our spanning tree of $X_1$, and the tree $T_{\setminus v}(X)$ rerooted at $l(X_2)$ forms our spanning tree for $X_2$. Clusters in $P_{-1}$ consist of single nodes, thus, they cannot split. The cluster at level $h$ consists of the entire graph. Even if edge failures occur, we do not split this cluster, but simply allow for the diameter of the cluster to grow.
        
        The following two lemmas bound the diameter and the number of the generated clusters.
        
        \begin{lemma}\label{lem: change diameter cluster}
            Let $X_i$ be any cluster at level $-1\leq i < h$. Regardless of the number of edge failures that have occurred, the diameter of $X$ is $\diam(X_i)\leq 2\sigma r_i$.
        \end{lemma}
        
         \begin{proof}
            If $X_i$ is one of the initial clusters, then the lemma holds trivially. Suppose $X_i$ was generated through the splitting of the initial cluster $X$. We know that the path from any node on $T(X)$ to $l(X)$ on $T(X)$ was at most $\sigma r_i$. Let $u_i$ be the node in $X_i$ that was closest to $l(X)$ on $T(X)$. For any two nodes $a$ and $b$ in $X_i$ we have $d(a, b)\leq d(a, u_i) + d(u_i, b)\leq 2 \sigma r_i$, where the last inequality follows from the fact that $u_i$ must be on the path from $a$, respectively $b$ to $l(X)$ on $T(X)$ and that no edge on the path from $a$ to $u_i$ and the path from $u_i$ to $b$ on $T(X)$ failed.
        \end{proof}
        
        \begin{lemma}\label{lem: linear increase cluster number}
            Let $X$ be a cluster and suppose that exactly $f$ edges fail. Then $X$ will split into at most clusters $X_1, \dots, X_l$ where $l\leq f+1$.
        \end{lemma}
        
        \begin{proof}
            Let $F$ be the set of edges that fail. The final clusters correspond to the connected components of $T(X)\setminus F$. Removing $j$ edges from a tree, leaves $j+1$ connected components. The set $T(X)\cap F \subseteq F$, hence $|T(X)\cap F|\leq f$ and $T(X)\setminus F$ consists of at most $f+1$ connected components.
        \end{proof}

        The initialization of the reclustering is identical to the initialization to update the shortest path tree. Say, node $u$ sends a message to $l(X)$ along $T(X)$ informing it of the failure. Each intermediate node $w$ updates its knowledge of $T_{\setminus w}(X)$ to exclude $T_{\setminus v}(X)$. If node $v$ is not in $X$ (that is, in case of weak diameter of $X$), then it chooses a node $w$ closest to it on $T(X)$ from $X\cap T_{\setminus v}(X)$ to become the new leader node and sends a message to $w$ to inform it of the reclustering and its new leadership role. Every node $x$ on the path from $v$ to $w$ includes its information about $T_{\setminus x}(X)$, so that $w$ and every node on the path from $v$ to $w$ knows its subtree in $T_{\setminus v}(X)$ rerooted at $w$.

        Next, we consider the root level. The spanning tree $T_X$ of this cluster corresponds to the shortest path tree rooted at $r$. This implies that before edge failure, $d(r, u)\leq \sigma\rho^h$ on $T_X$ for all nodes $u\in V$. If edge $e$ does not lie on $T_X$, then the diameter of the graph is at most $2\sigma\rho^h$. Therefore, in this case, we do not modify the root level. 
        
        If $e$ is on $T_X$, then we need to determine if there is a node whose distance to $r$ is larger than $\mathcal{O}(\sigma\rho^h)$, which would require an increase in the number of hierarchy layers. Define $V_1$ and $V_2$ as the partition of the vertices of $G$ defined by the connected components of $T_X\setminus \{e\}$, and assume $r\in V_1$. According to Lemma~\ref{lem: change diameter cluster} the diameter of $G[V_1]$ and $G[V_2]$ is at most $2\sigma\rho^h$. Let $x$ be the node with the maximum distance to $r$. If $x\in V_1$, $d(x,r)\leq 2\sigma \rho^h$, and all nodes are within distance $2\sigma\rho^h$ $r$. If $x\in V_2$, then the shortest path from $x$ to $r$ contains at least one edge that has an end point in $V_1$ and another in $V_2$. Let $e^*=\{a, b\}$ be the heaviest edge among these edges. Thus, $d(x,r) = d(x, a) + w(e^*) + d(b,r)$. Thus, we can determine whether we need to add additional layers to the hierarchy based on the weight of $w(e^*)$. In particular, if $w(e^*)> \sigma\rho^{h+1}-4\sigma\rho^h$, then we increase the number of layers to $h'$, where $h'$ is the smallest integer such that $\sigma\rho^{h'}> d(x,r)$. Note that this choice of $h'$ ensures that for all layers $i$ such that $-1\leq i \leq h'$ the diameter of a cluster at level $i$ is at most $2\sigma\rho^i$. 
        
        If $h'>h$, then the layers $h, h+1, \dots, h'-1$ consist of 2 clusters, namely $V_1$ and $V_2$. The splitting on level $h$ is performed in the same way that we split a cluster at a lower level whose spanning tree contains edge $e$. Also, we choose $\phi_h$ the same way that we detect the correct leader at the levels below (see Subsection \ref{subsec: updating directory path}). Levels $h+1,\dots, h'-1$ are simply copies of level $h$, and the directory path goes through the copy of the cluster that is part of the directory path at level $h$. When a cluster below the root level splits, then only the leader of $X_2$ sends a message to all nodes in $X_2$ to inform them of the cluster update. If the root level changes, then we require both $l(X_1)$ and $l(X_2)$ to send a message to all nodes in their cluster to inform them of the number of additional layers. This message will also cause the directory path nodes at levels $h-(i + \log_\rho(c'\sigma))+1, \dots, h' - i + \log_\rho(c'\sigma)$ to send a message to their special parent, so that this information is also extended to the additional layers. Level $h'$ consists of the entire graph and the leader is $r$. The spanning tree used for level $h'$ corresponds to the updated shortest path tree of node $r$. 
   
    \subsection{Updating the Directory Path and Special Parents}\label{subsec: updating directory path}
        We also need to ensure that the directory path and the special parent information are maintained. Suppose that cluster $X$ at level $i$ splits due to an edge failure. If $l(X)$ is not on the directory path, then neither $l(X_1)$ nor $l(X_2)$ will be. If $l(X)$ is on the directory path, then the one that contains the node $w$ that added $l(X)$ to $\phi$ will be. When $l(X)$ is informed about the edge failure, it can determine whether $l(X_2)$ is part of the directory path, because it remembers the node that added it to $\phi$, and it knows which nodes remain in $X_1$.
        
        Because $l(X_2)$ has no information about $\phi$ before the edge failure, node $l(X_1)$ sends a message to $v$, informing it whether $l(X_2)$ is part of the directory path or not. In case we are using a weak-sparse partition, node $v$ forwards the message to $l(X_2)$. If $l(X_2)$ is on the directory path, then the message contains the id of the directory path nodes that are on the level above and below, as well as the id of the node that added $l(X)$ to the directory path. Node $l(X_2)$ then sets its pointers to $\phi_{i+1}$ and $\phi_{i-1}$ and sends them a message so they update their pointers also. When $\phi_{i-1}$ and $\phi_{i+1}$ change their pointers to $l(X_2)$, they send a message to $l(X)$, so that it removes its outdated directory path pointers. If $l(X_2)$ is not part of the directory path, then the message $l(X)$ sends to $v$ simply informs $l(X_2)$ that it is not part of the directory path. The process of updating the directory path is displayed in Figure \ref{fig: updating the directory path}. After $l(X_2)$ has received $l(X)$'s message, and, if needed, updated the directory path, it broadcast a message to all nodes in $X_2$. The broadcast message informs the nodes in $X_2$ that their level $i$ leader has changed. Each node in $X_2$ will forward this information to all nodes in their $r_i$-neighborhood, so they can update their preprocessing information. 
        
        \begin{figure}[ht]
            \centering
            \includegraphics[height=3cm]{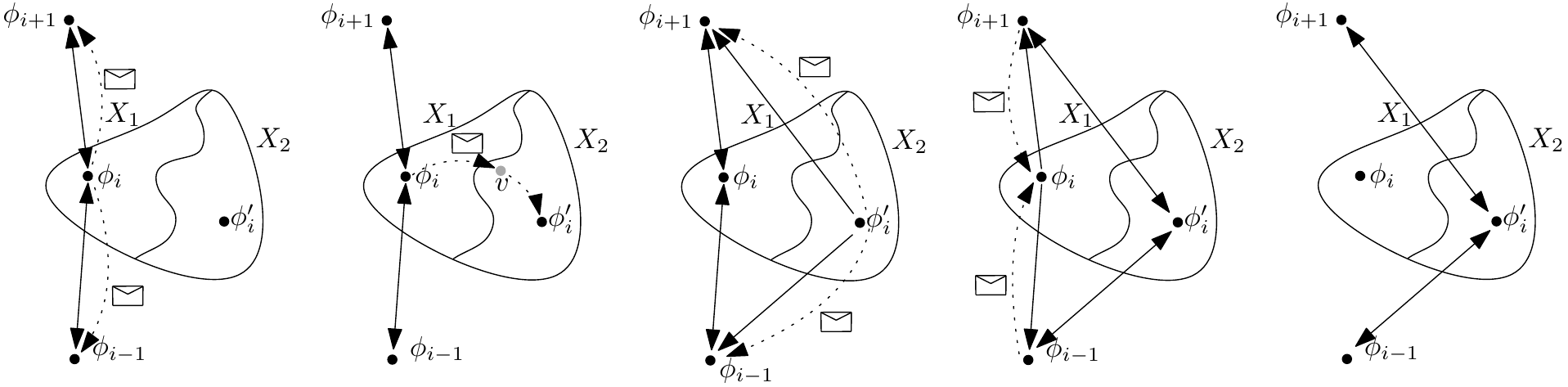}
            \caption{The steps of updating the directory path at level $i$: \\
            1) An edge failure causes the need to be updated from $\phi_i$ to $\phi_i'$. Node $\phi_i$ sends a message to $\phi_{i-1}$ and $\phi_{i+1}$ about the update that is about to happen.\\
            2) Node $\phi_i$ sends a message to $\phi_i'$ to inform it about the directory path. If we are using a weak sparse partition, node $v$ acts as an intermediate in the message transfer.\\
            3) $\phi_i'$ sets pointers to $\phi_{i-1}$ and $\phi_{i+1}$ and sends them a message so they update their pointers too.\\
            4) $\phi_{i-1}$ and $\phi_{i+1}$ change their pointer to $\phi_i'$ and send a message to $\phi_i$ to remove its pointers.\\
            5) $\phi_i$ removes its pointers to $\phi_{i+1}$ and $\phi_{i-1}$.}
            \label{fig: updating the directory path}
        \end{figure}
        
        For simplicity, we do not want two consecutive nodes on the directory path to be updated concurrently (when the modification is due to an edge failure). Therefore, if $\phi_i= l(X)$ needs to be replaced by $l(X_2)$, then $\phi_i$ will contact $\phi_{i-1}$ and $\phi_{i+1}$ to inform them of the update, before messaging $l(X_2)$. Neither of them will be able to initialize an update on their level until the update at level $i$ is complete. In case two subsequent nodes $\phi_i$ and $\phi_{i+1}$ attempt to initialize an update of the directory path simultaneously, (that is, $\phi_i$ sends a message to $\phi_{i+1}$ to initialize an update, but before $\phi_{i+1}$ receives this message it sends a message to $\phi_i$ to also initialize the an update of the directory path) then the node with the lower id will be allowed to do its update first.
        
        When we modify the directory path, we also need to update the special parent information of $l(X_1)$ and $l(X_2)$. This is done at the same time as we modify the directory path. When $l(X_1)$ detects that it is no longer part of $\phi$ it sends a message to its special parent at level $i'$ instructing it to remove the link to $l(X_1)$, and when $l(X_2)$ is informed that it is part of $\phi$ it also informs its special parent at level $i'$.
        
    \subsection{Message Fault Resilience}\label{sec: message fault resilience}
        Whenever an edge $e=\{u, v\}$ fails, any message in transit on $e$ is lost. To restore such messages, the two endpoints communicate to compare their lists of messages sent and received along the failed edge. Any message that does not appear on both lists will be resent.
    
        The mechanisms described above allow us to update the data structures during the execution of the distributed directory. For a detailed discussion on how the mechanisms affect concurrent publish, lookup and move operations and how we can ensure that each operation is still executed correctly, we refer the reader to Section \ref{app: Dynamic Setting}. A discussion on how to handle concurrent edge failures can be found in Section \ref{app: concurrent edge failures}.


\section{Analysis of Algorithm} \label{sec:analysis}
    The analysis of our algorithm in the absence of failure is similar to the analysis of the Spiral~\cite{Sharma2014distributed} and can be found in Section \ref{sec: Analysis without faults}.

    \subsection{With up to f Faults}\label{subsec: Analysis with Faults}
        Operations that occur concurrently to, or immediately after an edge failure can endure additional costs because the special parent information got updated too slow, or because the directory path got stretched. We call such operations \emph{transient operations}. A discussion of transient operations can be found in Section \ref{app: transient operations}. Operations that are not transient are called \emph{normal operations}.  
        
        \subsubsection{Normal Operations}
            For two consecutive nodes that are added to the directory path by a move operation after the last edge failure, one of the two cases discussed in Lemma \ref{lem: distance consecutive nodes} must apply. Because the diameter of a clusters can increase due to the edge failure (Lemma \ref{lem: change diameter cluster}), the length of the directory path after $f$ edge failures can be longer than before any faults occurred.
            
            The next lemma helps to bound also the cost of a publish operation after failures.
            
            \begin{lemma}\label{lem: length directory path after failure}
                Suppose that since the last edge failure, a publish or move operation has rebuilt the directory path up to level $i$. Then the length of the directory path up to level $i$ is at most $\mathcal{O}(\rho^i\sigma)$.
            \end{lemma}
            
            \begin{proof}
                From the discussion in Lemma \ref{lem: distance consecutive nodes}, the distance between two consecutive nodes can be bounded by $d(\phi_j, \phi_{j+1})\leq d(\phi_j, u) + d(u, w) + d(w, \phi_{j+1})$.
                Because $w$ found $\phi$ at level $j+1$, we have $d(u, w) \leq r_{j+1}$. 
                After the repair of the data structure, the diameter of cluster $X$ at level $i$ is at most $\diam(X)\leq 2\sigma r_i$. Thus, we obtain
                \begin{align*}
                    &\sum_{j=-1}^{i-1}d(\phi_j, \phi_{j+1})
                    \leq \sum_{j=-1}^{i-1} 2\sigma(\rho^i + \rho^{i+1})+ \rho^{i+1}
                    \leq \frac{(\rho^{i+1}-1)(2\sigma(\rho+1) + \rho)}{(\rho-1)\rho} = \mathcal{O}(\sigma \rho^{i}).
                \end{align*}
            \end{proof}

            Before we analyze the cost of the lookup and move operation, we bound the number of clusters affected by a single edge failure. 
    
            \begin{lemma}\label{lem: number cluster to recluster}
                The total number of clusters that need to be reclusterd due to a single edge failure is at most $h$ in a strong sparse partition, and at most $Ih$ in a weak sparse partition.
            \end{lemma}

            \begin{proof}
                A cluster $X$ splits due to an edge failure $e=\{u, v\}$, if $T(X)$ contains $e$. In a strong sparse partition, all nodes on $T(X)$ belong to $X$. Thus, edge $e$ can only be part of $T(X)$, if $X$ contains both $u$ and $v$. Because every node appears in exactly one cluster on each level of the hierarchy, at most $h$ clusters need to be reclusters upon a single edge failure if we use a strong sparse partition.
                
                In a weak sparse partition, the spanning tree of cluster $X$ might include nodes that are not part of $X$. However, the $r_i$-neighborhood of a node $u$ intersects at most $I$ clusters, for all $i$. As any cluster whose spanning tree contains $e$ also contains $u$ and $v$, there can be at most $I$ clusters on any level that need to be reclustered due to the failure of $e$.
            \end{proof}
            
            Furthermore, additional layers may need to be added to the directory to account for the increase in the diameter caused by the edge failures. 
            
            The analysis of the lookup and move operation after edge failures have occurred is very similar to the analysis in the absence of failures. Therefore, we have moved the proofs of the following theorems to Section \ref{sec: proofs}.
            
            \begin{restatable}[]{theorem}{lookupstrong} \label{lem: lookup with faults strong}
              Suppose we are using a strong sparse partition and that $f$ edge failures have occurred. After updating our data structures, a lookup operation finds the token with cost that is $\mathcal{O}(\sigma^2 (I+f)\rho)$ factor from optimal.
            \end{restatable}
            
            \begin{restatable}[]{theorem}{lookupweak}\label{lem: lookup with faults weak}
              Suppose we are using a weak sparse partition and that $f$ edge failures have occurred. After the clusters, directory path, preprocessing information, and special parents have been updated, a lookup operation finds the token with cost that is $\mathcal{O}(\sigma^2fI \rho)$ factor from optimal.
            \end{restatable}

            \begin{restatable}[]{theorem}{movestrong}\label{lem: cost move with failures strong}
              Consider a sequence $S$ of move requests $S=s_1,\dots, s_q$ that are all issued after the $f^{\text{th}}$ edge failure and which are executed sequentially.
              The total cost of the move operations in $S$ is a $\mathcal{O}(h'\sigma\rho((I+f) + \sigma)$ factor from optimal in a strong sparse partition (for sufficiently large $S$). 
            \end{restatable}
            
            \begin{restatable}[]{theorem}{moveweak}\label{lem: cost move with failures weak}
                The total cost of the move operations in $S$ is a $\mathcal{O}(h'\sigma\rho(fI + \sigma)$ factor from optimal in a weak sparse partition (for sufficiently large $S$).
            \end{restatable}

\section{Detailed Description of Sparse Partitions}\label{app: sparse partition}
        Our directory is based on a partition hierarchy of graph $G = (V,E,w)$. A partition of $G$ is a collection of disjoint clusters, where each cluster is a set of nodes, whose union is $V$. We will use {\em sparse partitions} which limit the diameter of each cluster and also limit the number of clusters within a certain distance.
        
        The distance $d_G(u,v)$ between any two nodes $u,v \in V$ is the length of the shortest path between the two nodes in $G$. The diameter of a graph is $\diam(G) = \max_{u,v \in V(G)} d(u, v)$. For any set of nodes $X \subseteq V$, let $G[X]$ denote the subgraph of $G$ induced by $X$. There are two ways to measure the diameter of $X$: (i) {\em weak diameter}, $diam_G(X)$, which considers all possible shortest paths in $G$ that may also use nodes outside $X$; (ii) {\em strong diameter}, $diam_{G[X]}(X)$, which considers only paths within $X$. Note that strong diameter implies that $G[X]$ is connected, while weak diameter does not guarantee connectivity. We denote the $r$-neighborhood of a node $u \in V$ in $G$ as $N_{G,u}(r) = \{v \in V: d_G(u,v) \leq r\}$.
    
        A {\em $(r,\sigma,I)$-weak (strong) sparse partition} of $G$ satisfies two properties:
        \begin{itemize}\setlength\itemsep{0em}
            \item[(i)] each cluster has weak (strong) diameter at most $r \sigma$, and
            \item[(ii)] the $r$-neighborhood of each node $u \in V$ intersects at most $I$ clusters.
        \end{itemize}
        A {\em $(\sigma, I)$-weak sparse partition scheme} is a procedure that gives a $(r,\sigma,I)$-weak (strong) partition for any $r > 0$. Jia {\em et al.} \cite{jia2005universal} give a $(\mathcal{O}(\log n), \mathcal{O}(\log n))$-weak sparse partition scheme for an arbitrary metric space which is also suitable for general graphs. Filtser \cite{filtser2020scattering} gives a $(\mathcal{O}(\log n), O(\log n))$-strong partition scheme for general graphs based on the clustering technique by Miller {\em et al.} \cite{miller2013parallel}. There are $(\mathcal{O}(1), \mathcal{O}(1))$-partition schemes for special network topologies such as for low doubling-dimension and fixed minor-free graphs~\cite{jia2005universal,filtser2020scattering}.
        
        By picking any of the above partition schemes, we can construct a partition hierarchy ${\mathcal P}$. Let $\rho > 1$ be a locality parameter for the cluster diameter. The hierarchy consists of multiple levels by exponentially increasing $\rho$ at each level. Let $D$ be the diameter of $G$ and $h = \lceil \log_\rho D \rceil$. At level $i$, $0 \leq i \leq h$, let $P_i$ be a $(r_i,\sigma,I)$-sparse partition of $G$, where $r_i = \min\{D, \rho^i\}$. Let $P_{-1}$ be the partition where each node of $V$ is a cluster by itself. For convenience, assume that $r_{-1} = 0$. Assume that $P_h$ has a single cluster which is the whole of $V$. The hierarchy is ${\mathcal P} = \{P_{-1}, P_0, \ldots, P_h\}$.


\section{Handling Concurrent Edge Failures}\label{app: concurrent edge failures}
    Let $e_1=\{a, b\}$ and $e_2=\{c, d\}$ be two edges that fail concurrently. We first describe how to maintain the cluster data structure. 
        
    \subsection{Updating the Cluster Data Structure}
        Let $X$ be a cluster whose spanning tree contains both $e_1$ and $e_2$. We assume $d(a, l(X)) < d(b, l(X))$ and $d(c, l(X)) < d(d, l(X))$. There are two cases to consider $e_2\notin T_{\setminus b}(X)$, and $e_2\in T_{\setminus b}(X)$. 
   
        In the first case, three clusters will be generated $X_1 = X\setminus ((X\cap T_{\setminus b}(X))\cup (X\cap T_{\setminus d}(X)))$, $X_2= X\cap T_{\setminus b}(X)$, and $X_3= X\cap T_{\setminus d}(X)$. If $e_2\in T_{\setminus b}(X)$, then the path connecting $a$ and $l(X)$ is not affected by the failure of $e_2$, and the path connecting $c$ and $l(X)$ is not affected by the failure of $e_1$. Therefore, the mechanisms that handle the failure of $e_1$ do not interfere with the mechanisms that handle the failure of $e_2$.  
        
        Consider the second case. The three clusters generated in this setting are $X_1 = X\setminus (T_{\setminus b}(X)\cap X)$, $X_2 = (X\cap T_{\setminus b}(X))\setminus (X\cap T_{\setminus d}(X))$, and $X_3 = X\cap T_{\setminus d}(X)$. When edge $e_2$ fails, node $c$ will send a message $m_2$ toward $l(X)$ along $T(X)$ to inform it of the failure. There are two cases to consider. Either $m_2$ traverses $e_1$ before $e_1$ fails, or it does not. If $m_2$ traverses $e_1$ before it fails, then node $b$ will be aware of the failure of $e_2$ at the time $e_1$ fails. Therefore, it chooses $l(X_2)$ to be in $(X\cap T_{\setminus b}(X))\setminus (X\cap T_{\setminus d}(X))$ and informs the new leader node and all other nodes on the path to the leader node of the correct spanning tree. We assume that messages are delivered in a FIFO fashion. Hence, if $m_2$ traversed $e_1$ before $e_1$ failed, then $m_2$ will reach $l(X)$ before $m_1$, the message generated by $a$ to inform $l(X)$ of the failure of $e_1$. When $l(X)$ is informed about the failure of $e_2$ it determines if it needs to update the directory path so that $X_3$ is part of the directory path, and similarly, it determines if $X_2$ needs to be part of the directory path when it receives message $m_1$.
        If $e_1$ fails before $m_2$ traverses it, then $l(X)$ will only be informed of the failure of $e_1$ and act according to this failure. When the message $m_2$ reaches node $b$, it knows $l(X_2)$ and forwards $m_2$ to $l(X_2)$. Node $l(X_2)$ processes this message like a single edge failure. 
        
        The cases $e_1\in T_{\setminus d}(X)$ and $e_1\not\in T_{\setminus d}(X)$ are analogous to the ones discussed above.
        
        If we have more than two concurrent edge failures, then the mechanisms are very similar to above. In general, a leader node can be informed of one or several edge failures that are either in the same or in different subtrees. If the edge failures are in different subtrees then the leader node handles them in the same way as it handles the two edge failures discussed in case one above. If it is informed of several edge failures in the same subtree, then it handles them sequentially, like in the second case above. It can also be that there are multiple edge failures in the same subtree, but the leader node is only aware of some of them. As for the case with two edge failures, it must be that if the leader node $l(X)$ does not get informed of an edge failure $e$, then it must get informed of an edge failure $e'$ such that $T_e(X)$ is a subtree of $T_{e'}(X)$. In this case, the leader node of the partitioned off cluster, generated through the failure of $e'$, will handle the failure of $e$.
    
    \subsection{Updating the Shortest Path Trees}
        To maintain the shortest path trees in the presence of concurrent edge failures, we note that King's algorithm~\cite{king1999fully} can handle multiple concurrent edge failures. If the root node of the tree is aware of all edge failures before it starts the repair algorithm, it will handle all failures at the same time. Otherwise, if the root node is informed of several edge failures one after the other, it can simply trigger the algorithm multiple times. In the second case, the result returned by iterations other than the last might contain the faulty edges that the root is unaware of. The final iteration will return a tree with no failed edges.
        
    \subsection{Updating the Directory Path and Preprocessing Information}
        The directory path is maintained in a similar way as for single edge failures. Every time a leader node $l(X)$ is informed about an edge failure if determines if the partitioned off cluster should be part of the directory path and informs the cluster accordingly. If a cluster $X$ that recently split from another cluster splits, and $l(X)$ has not yet been informed if it is part of the directory path, then $l(X)$ waits to be informed, before it sends a message to the partitioned off cluster.
        
        For the preprocessing information, every node simply informs its $r_i$-neighborhood every time it is informed about a cluster split at level $i$. If there are multiple edge failures, and a node $u$'s cluster changes several times, then it will also inform its $r_i$-neighborhood about a cluster update several times.
        

\section{Adding Fault Tolerance to the Directory}\label{app: Dynamic Setting}
    In this Section, we discuss how updating the data structures as described in Section \ref{sec: Responding to Edge Failures} affects the publish, move, and lookup operations. For the analysis of move and lookup operations, we will split the operation into two phases: searching for the directory phase and following the directory path to the token phase. We analyze the effects of an edge failure on either of these two phases independently.
    
    We first discuss how the modifications to the publish, move, and lookup operation described at the beginning of Section \ref{sec: Responding to Edge Failures} help us. 
    
    The first modification is needed so that we can update the directory path when an edge failure occurs. A directory path node $\phi_i$ whose cluster $X$ splits due to the edge failure needs to know which node added it to the cluster, to be able to determine if it should remain part of the directory path, or if the leader node of the partitioned off cluster should replace it. 
   
    The second and third modifications ensure that the correctness and performance of a search for the directory path during a lookup or move operation are maintained. When a node $w$ searches for the directory path, it determines the leader nodes it needs to contact at each level. Node $w$ needs to contact cluster $C_i(x)$ when searching level $i$, if $x$ is in $w$' $r_i$-neighborhood. If the $i^{\text{th}}$ level cluster of $x$ changes due to an edge failure, then it can happen that $w$ contacts the wrong cluster if sends a message to $l(C_i(w))$ before being informed of the update. For this reason, the message that $w$ sends to a cluster leader to query about the directory path includes a list of all nodes that $w$ believes to be part of the cluster. If one of the nodes in the list, say node $x$, is no longer part of the cluster, then the contacted leader node will inform node $w$ to wait for a message from $x$ about a cluster change. Once $w$ is informed about $x$'s new cluster, it will also contact the leader node of that cluster. Only then can $w$ move on to search the directory path at the next higher level.

    The cost of contacting the level $i$ cluster leader of a node $x$ in $w$'s $r_i$-neighborhood is at most $d(w, l(X))\leq d(w, x) + d(x, l(X)) \leq r_i + \diam(X)$. According to Lemma \ref{lem: change diameter cluster} the diameter of a level $i$ cluster in the presence of edge failures is at most $\diam(X)\leq 2\sigma r_i$. Suppose cluster $X$ splits due to an edge failure and the leader of $x$ changes. Then the distance between $x$ and its outdated leader node $l(X)$ is no longer bounded by $\diam(X)$. Before node $w$ contacts a leader $l(X)$ of node $x$ on level $i$, it checks, using its shortest path tree, if $d(w, l(x))\leq r_i + 2\sigma r_i$. If $d(w, l(X))> r_i + 2\sigma r_i$, then it must be that the level $i$ cluster of $x$ changed. Therefore, node $w$ does not even need to contact $l(X)$, but can simply wait for a cluster update info by node $x$.

    \subsection{Edge Failure during Publish Operation}
        Suppose node $w$ issues a publish operation and a failure of edge $e=\{u, v\}$ occurs before the publish operation reaches the root level. If edge $e$ is not part of $w$'s shortest path tree and none of the clusters that contain node $w$ is split, then the publish operation is not affected by the edge failure. 
        
        In the case where edge $e$ lies on $w$'s shortest path tree, then when $w$ is informed of the edge failure by $u$ or $v$ it updates its shortest path tree and continues to build the directory path up to the root. We note that $w$ uses its shortest path tree to contact its leader node at each level efficiently. However, the shortest path tree has no influence on the nodes that make up the directory path. Hence, the update of the shortest path tree has no further impact. 
        
        Let $X=C_i(w)$, and suppose $X$ is split due to the failure of $e$. If the level $i$ leader node of $w$ does not change due to the splitting of $X$, then the publish operation is not affected by the failure. If $w$'s leader node at level $i$ changes due to the edge failure, and $w$ is informed of the change before adding $l(X)$ to the directory path, then $w$ will simply add its new level $i$ leader to the directory path. If the directory path has already been built to level $i$, then $l(X)$ detects that $w$, is no longer part of its cluster. Because $l(X)$ knows that $w$ was the node that added $l(X)$ to the directory path, it will initialize the process to update the directory path as described in Section \ref{subsec: updating directory path}.
        
    \subsection{Edge Failure on the Shortest Path Tree during a Move or Lookup Operation}\label{subsec: move shortest path tree}
        Suppose node $w$ issues a move or lookup operation and that before it receives the token an edge failure of edge $e=\{u, v\}$ occurs. If $e$ lies on $w$'s shortest path tree then the distances from $w$ to the nodes in the subtree below $e$ might increase due to the edge failure. For the search of the directory path phase, this means that some nodes need only be contacted at a higher level than before the edge failure. Assume node $w$ is informed of the edge failure while it searches level $i$ for the directory path. Let $x$ be a node such that the distance between $w$ and $x$ changes due to the edge failure. Define $d$ as the distance between $w$ and $x$ before the failure of $e$, and $d'$ as the new distance. Assume $r_{j-1} < d \leq r_{j}$ and $r_{k-1} < d' \leq r_k$. Because an edge failure cannot decrease the distance between two nodes, it must be that $j\leq k$. If $j > i$ and $k>j$ then $w$ continues its search for the directory path up to level $j$ as it would have without the edge failure. When it searches level $j$ it does not contact $C_j(x)$ unless there is another node in this cluster that is in $w$'s $r_j$-neighborhood. The first time it contacts $x$'s leader node is at level $k$ (unless $w$ contacts $x$'s leader due to a different node in the cluster). If $j\leq i$ and $k> i$, then $w$ does not contact $C_i(x)$, unless it already did so before being informed about the edge failure, or because there is a different node in $w$'s $r_i$-neighborhood that belongs to cluster $C_i(x)$. The first time it contacts $x$'s leader node again is at level $k$. If $j\leq i$ and $k\leq i$ then $w$ will contact cluster $C_i(x)$ during its search of level $i$ and will continue to contact $x$'s cluster at every level thereafter until it finds the directory path.
            
        For the following the directory path downward phase, node $w$ will still contact the nodes as they appear on the directory path. However, the price it needs to pay to contact these nodes might increase.
        
    \subsection{Edge Failure during Move Operation}
        
        \subsubsection{While Searching for the Directory Path}
            Suppose node $w$ issues a move operation and before it finds the directory path an edge failure of an edge $e=\{u, v\}$ occurs. Let $\mathcal{X}$ denote the set of clusters that split due to the failure of edge $e=\{u, v\}$. If node $w$ does not contact a leader node of a cluster in $\mathcal{X}$ during its search for the directory path, then the search phase of the move operation is not affected by the edge failure. 
            
            Suppose $w$ needs to contact a cluster leader $l(X)$ of a cluster $X$ in $\mathcal{X}$ during its search phase. Then there is a node $x$ in $X$ such that $d(w, x)\leq r_i$ after the edge failure, and $X$ is the level $i$ cluster of $x$ before the edge failure. If $x$ remains connected to $l(X)$ on $T(X)$, then the level $i$ cluster leader of $x$ does not change. In this case, $w$ will not be informed of the splitting of cluster $X$, unless there is another node $y$ in its $r_i$-neighborhood that was part of cluster $X$ and got disconnected from $l(X)$ on $T(X)$ due to the edge failure. 
            
            Suppose the level $i$ leader node of $x$ changes due to the edge failure. In this case, node $x$ will get informed of the cluster change through a broadcast message from its new level $i$ leader. When this occurs, node $x$ sends a message to all nodes in its $r_i$-neighborhood to inform them of the update. There are two cases to consider: Either node $w$ is informed about the cluster change prior to messaging $l(X)$, or it contacts $l(X)$ before being informed about the cluster change of $x$. In the first case, node $w$ determines whether it still needs to contact $l(X)$, that is, $w$ determines whether there is another node in its $r_i$-neighborhood that belongs to cluster $X$ at level $i$. If there is such a node, then $w$ contacts both $l(X)$ and $x$'s new leader node to search for the directory path. Otherwise, it will only contact $x$'s new leader node at level $i$. Note that node $x$ is only informed of the cluster change after the new leader node knows if it is part of the directory path or not. Therefore, when $w$ contacts $x$'s new level leader node, this leader is guaranteed to know if it is part of the directory path or not. 
            
            In case node $w$ contacts $l(X)$ before being informed about the cluster update, we distinguish two cases. Namely, whether $l(X)$ was informed about the edge failure and, therefore, the splitting of $X$ before or after $w$ contacts $l(X)$. If $w$ contacts $l(X)$ before $l(X)$ is informed about the edge failure, then $w$ does not need to contact $x$'s new level $i$ cluster, because the information $w$ receives from $l(X)$ is also valid for $w$'s new cluster.
            
            However, if $l(X)$ is already informed of the splitting of $X$ at the time $w$ contacts $l(X)$, then $l(X)$ will inform $w$ that it is not the correct leader to contact, since $x$ is no longer part of its cluster. In this case, node $w$ waits for the cluster update message from $x$, and then contacts $x$'s new leader. Node $w$ can contact the remaining level $i$ leader nodes that it needs to contact while waiting for $x$'s update. However, $w$ cannot move on to search the next level before contacting the correct cluster of $x$'.
            
            The construction of the new directory path that is done while searching for the old directory path is not affected by edge failures. Node $w$ adds a node $l_i(w)$ to the directory path if it does not find the directory path at level $i$. As described above, after node $w$ finishes searching a particular level for the directory path, it does not need to go back to that level even if some of the clusters that $w$ contacted during the search at level $i$ split due to an edge failure. Therefore, if node $l_i(u)$ is added to the directory path, then the $r_i$-neighborhood of level $w$ does not intersect the level $i$ cluster that is part of the directory path. If the level $i$ cluster that contains node $u$ splits due to an edge failure, then the leader node will determine if the directory path needs to be modified due to the split and initialize the update process accordingly.

            If an edge failure causes the root level to split and additional layers to be added to the directory, we can handle it the same way as any other cluster splits, because the clusters at level $h+1, \dots, h'-1$ are simply copies at level $h$. Hence, once level $h$ is done updating so are the clusters at levels $h+1, \dots h'$.

       \subsubsection{While Following the Directory Path Downward} 
            Suppose node $w$ has found a cluster $X$ whose leader node $l(X)$ is on $\phi$. Then it follows the downward pointers of the nodes on $\phi$ to reach the node that holds the token. We again denote the set of clusters that split due to the edge failure of $e$ by $\mathcal{X}$. If the move operation does not traverse any clusters from $\mathcal{X}$ in its downward phase, then the downward phase is not affected by edge failure. Suppose that the move operation traverses a cluster $X$ in $\mathcal{X}$. If the splitting of $X$ does not cause a modification of the downward phase, then again the downward phase of the move operation is not affected by the edge failure. Thus, suppose that the splitting of $X$ causes a modification of the directory path at level $i$.
            
            Suppose that $l(X)$ has already realized the need to modify the directory path, but the modification is not finished when the move operation reaches node $\phi_{i+1}$. There are two cases to consider. Either, cluster $X$ is waiting for a modification of the directory path at the level above or below to finish, or node $l(X)$ has already sent a message to a node in the partitioned off cluster to initialize the modification. In the first case, the move operation can simply traverse the directory path from level $i+1$ to level $i$ to level $i-1$ as all the pointers are still intact. When traversing $\phi_i$, it removes the directory path pointers. This will prevent $l(X)$ from initializing a modification of the directory path. In the case where $\phi_i$ has already initialized the process of modifying it will have sent a message to $\phi_{i+1}$ to inform it of the update and to ensure $\phi_{i+1}$ does not initialize an update simultaneously. In this case, $\phi_{i+1}$ will not forward the move message until the update is completed. The pointers to $l(X)$ will be removed in the modification of the directory path operation, and the pointers of the updated $\phi_i$ will be removed by the move message.
            
            We do not need to consider the situation where the modification has finished before the move has reached $X$, or where $l(X)$ is informed of the failure after the move operation traverses $\phi_i$. 
            
    \subsection{Edge Failure during Lookup Operation}
    
        \subsubsection{While Searching for the Directory Path}
            The searching for the directory path phase of the lookup operation is similar to the searching for the directory path phase of the move operation, as discussed in Section \ref{sec: Directory Scheme}. The differences are that the lookup operation also uses the information provided by special parents and that the lookup operation does not modify the directory path. We thus only discuss the impact an edge failure has on the special parent information. 
            
            Suppose node $w$ issues a lookup operation and before it receives a copy of the object, an edge failure of edge $e=\{u, v\}$ occurs. 
            If $w$ queries node $l(X)$ about the directory path and is informed that $l(X)$ is the special parent of a node $l(X')$ on the directory path, then it can happen that while $w$'s lookup message traverses the link to $l(X')$, the edge failure of $e$ causes a modification of the directory path. 
            If $l(X')$ is no longer part of the directory path, when $w$'s lookup message reaches $l(X')$, then $w$ will go back to $X$'s level and continue its search for the directory path there. 
           
        \subsubsection{While Following the Directory Path Downward}
            Once node $w$ has found a node $l(X)$ that is part of the directory path, $l(X)$ forwards the request downward along the directory path to the node that holds the token. Unlike the move operation, the lookup operation does not need to wait for potential updates of the directory path because it does not need to modify the directory path. Suppose that the directory path needs to be modified at level $i$ due to the edge failure, and the lookup operation needs to contact $\phi_i$ in its downward phase. If the directory path is updated before the lookup operation reaches level $i$ then the lookup will follow the updated directory path. If the directory path is not updated when the lookup operation reaches level $i$, then $\phi_{i+1}$ will still have a pointer to $\phi_i$ and $\phi_{i}$ will still have a pointer to $\phi_{i-1}$. Therefore, the lookup message can traverse these links, even if $\phi_i$ has already initialized the process of updating the directory path at level $i$.
            
            Note that when we update the directory path at level $i$, then the node that is to be replaced removes its links only once it is informed by $\phi_{i+1}$ and $\phi_{i-1}$ that the new $\phi_i$ has been added to the directory path. Because we assume messages to delivered FIFO, we are guaranteed that even if an update at level $i$ has been initialized, if $\phi_{i+1}$ has not updated its pointer, then the outdated $\phi_i$ will also still have a pointer to $phi_{i-1}$. Therefore, the lookup message does not need to be sent back up to $\phi_{i+1}$.


\section{Transient Operations}\label{app: transient operations}
    Our failure mechanisms ensure that after the reclustering we can bound the diameter of any cluster. However, the distance between two nodes that are not in the same cluster can only be bounded by $D'$ the diameter of the graph $G\setminus\{e\}$, even if the two nodes were close before the edge failure. Consider, for example, a cycle graph on $n$ nodes with uniform edge weights. When an edge $e'=\{a,b\}$ fails in this graph, then the distance between $a$ and $b$ will be $\diam(G\setminus\{e\}) = n-1$ in $G\setminus\{e\}$.

    Suppose node $w$ issues a lookup operation and before it finds the directory path, an edge failure of edge $e=\{u, v\}$ occurs. Suppose node $x$ was the last node to modify the directory path at level $i$. That is $\phi_i=l_i(x)$. If there is a node $y$ such that after the edge failure we have $d(w, y)\leq r_i$ and $C_i(y) = C_i(x)$ then our repair mechanisms ensure that $w$ will find the directory path at level $i$ at the latest, even if the directory path needs to be modified at level $i$ due to the failure of $e$. 
    
    When node $w$ searches for the directory path, it asks the leader nodes not only if they are part of the directory path, but also if they are the special parent of a node on the directory path. Suppose that the directory path needs to be modified at level $i$ due to the edge failure of $e$. Let $l(X)$ be the leader node at level $i$ before the edge failure, and $l(X')$ be the new leader node at level $i$. Although $l(X)$ and $l(X')$ are in the same cluster before the edge failure, this does not imply that there are in the same cluster at level $i' = i + \log_\rho(c'\sigma)$, the level where their special parent resides. If $C_{i'}(l(X))\neq C_{i'}(l(X'))$ then, when $l(X')$ is added to the directory path, it also needs to inform $l_{i'}(l(X'))$, and similarly when $l(X)$ is removed, it needs to inform $l_{i'}(l(X))$. Note that there may be some time delay between the edge failure and the update of the directory path. If node $w$ contacts $l_{i'}(l(X'))$ after the edge failure, but before $l_{i'}(l(X'))$ is informed that it is the special parent of a node on the directory path, then $w$ does not find a link to the directory path at level $i'$, even though the distance between $w$ and the node that holds the token is such that without the edge failure it would have found the directory path at level $i'$. 
    
    The length of the directory path can also increase due to an edge failure, even after we modify the directory path. Let $\phi_i$ and $\phi_{i-1}$ be two consecutive nodes on the directory path that are the leaders of the clusters $X_i$ and $X_{i-1}$. In Lemma \ref{lem: distance consecutive nodes}, we considered two cases for bounding the distance between two consecutive nodes on the directory path. In the first case, both nodes were added by node $u$, because it did not find the directory path at the corresponding levels. In this case, clusters $X_i$ and $X_{i-1}$ must both contain node $u$, and we have $d(\phi_i, \phi_{i-1})\leq d(\phi_i, u) + d(u, \phi_{i-1})$. As we show in Lemma \ref{lem: change diameter cluster}, the diameter of a cluster at level $i$ after any number of edge failures is bounded by $2\sigma r_i$. Once we modify the directory path, after the edge failure, we are guaranteed that $\phi_i$ and $\phi_{i-1}$ are the leaders of the clusters that contain node $u$ at level $i$, respectively $i-1$. Therefore, in this case, the distance between $\phi_i$ and $\phi_{i-1}$ after updating the directory path can be bounded by $d(\phi_i, \phi_{i-1})\leq 4 \sigma r_i$. 
    
    The second case that we considered to bound the distance between $\phi_i$ and $\phi_{i-1}$ is that node $u$ added $l_{i-1}(u)$ to the directory path and found the directory path at level $i$, because $X_i$ contains some node $w$ such that at the time $u$ searches for the directory path $d(u, w)\leq r_i$. In this case, we bounded the distance between $\phi_i$ and $\phi_{i-1}$ by $d(\phi_i, \phi_{i+1})\leq d(\phi_i, u) + d(u, w) + d(w, \phi_{i+1})$. After updating the directory path after an edge failure, we again have an upper bound for the first and the third term in the sum. However, as discussed, the upper bound we can give on $d(u, w)$ is the diameter of the generated graph. This implies that in the worst case, the directory path will have length $h'D'$ after an edge failure, where $D'$ is the diameter of the graph that remains after the edge failure.
    
    The lookup and move operation both need to traverse part of the directory path to reach the node that holds the token. Therefore, both of these operations might have to pay this cost. However, a move operation builds a new directory path and deletes the old one. Our reclustering and updating of the directory path ensures that a move operation that is issued after the edge failure occurred will build a new directory path in which the consecutive nodes $\phi_i$ and $\phi_{i-1}$ have distance at most $d(\phi_i, \phi_{i+1})\leq d(\phi_i, u) + d(u, w) + d(w, \phi_{i+1})\leq 2\sigma (r_i + r_{i+1}) + r_{i+1}$ where the last inequality follows from Lemma \ref{lem: change diameter cluster} and $d(u, w)\leq r_i$. Therefore, we can bound the length of the directory path from level $-1$ to level $i$ after an edge failure occurred, as soon as at least one move operation issued after the failure reached level $i$.


\section{Proofs Subsection \ref{subsec: Analysis with Faults}}\label{sec: proofs}
    
    \lookupstrong*
    
    \begin{proof}
        Assume that the lookup operation found the directory path at level $i$ and the directory path up to level $i$ adheres to the bound of Lemma \ref{lem: length directory path after failure}. Our analysis is identical as the analysis in Lemma \ref{lem: cost lookup without failure}. When we sum the cost of the search up to level $i$, we need to account for the increase in the cluster diameter and the number of clusters that need to be contacted at each level (Lemma \ref{lem: change diameter cluster} and Lemma \ref{lem: number cluster to recluster}). When node $u$ contacts a cluster leader node $l(X)$ at level $i$, the distance between $u$ and $l(X)$ is at most $d(u, l(X))\leq r_i + 2\sigma r_i\leq c\sigma r_i$ for some $c\geq 1 + \sigma$. We can thus bound the cost of the search by summing over all clusters that need to be contacted over all levels:
        \begin{align*}
            \text{cost search up to level $i'$} \leq \frac{c\sigma (I+f)(c' \sigma \rho^{i+1} -1)}{\rho -1}.
        \end{align*}
        Therefore, the upward phase of the lookup operation using a strong sparse partition has cost at most $\mathcal{O}(\sigma^2 (I+f)\rho^i)$. The cost of the downward phase is given by Lemma \ref{lem: length directory path after failure}. Summing the two costs, we see that the cost of the entire lookup operation is $\mathcal{O}(\sigma^2 (I+f)\rho^i)$ when using a strong sparse partition. Because the lookup operation found the directory path at level $i$ it must be that the optimal cost is at least $\sigma\rho^{i-1}$ as explained in Lemma \ref{lem: cost lookup without failure}. Therefore, the cost of the lookup operation is within a factor of  $\mathcal{O}(\sigma^2 (I+f)\rho)$ from optimal for strong sparse partitions.
    \end{proof}
    
    \lookupweak*
    
    \begin{proof}
        The proof is identical to Theorem \ref{lem: lookup with faults strong}, except that after $f$ edge failures $P_i(u)\leq fI$, according to Lemma \ref{lem: number cluster to recluster}. Thus, the lookup will need to visit up to $fI$ clusters on each level.
    \end{proof}
    
    \movestrong*
    
    \begin{proof}
        We use the same notation as in Lemma \ref{lem: cost move without failure}, with the exception that for each $i$, we define $s_{i_0}$ to be the last move operation prior to $S$ that reached level $i$. If no such operation exists $s_{i_0}$ is the initial publish operation otherwise. We note that $s_{i_0}$ can occur before or concurrently with the last edge failure. The bound for the optimal cost of serving the request in $S$ is identical to Lemma \ref{lem: cost move without failure}.
        
        To analyze the cost of our algorithm, we need to be a little bit more careful because operation $s_{i_1}$ could be a transient operation for $0\leq i\leq h'$. For the upward phase of the move operation the cost of a transient and a normal operation is identical. The cost analysis of the upward phase is hence identical to Lemma \ref{lem: cost move without failure} except we need to adjust for the change in diameter and the number of clusters that need to be contacted. We obtain $C(S_i)\leq|S_i|c\sigma r_i((I+f)+\sigma)$ for some constant $c$.
        
        For the downward phase, a transient operation might encounter two consecutive nodes $\phi_k$ and $\phi_{k-1}$ on the directory path with worst case distance $d(\phi_k, \phi_{k-1}) = \diam(G') = D'$, where $G'$ is the graph we obtain after the $f$ edge failures. However, for each level only the very first move operation that traverses the directory path downward can encounter this cost, as subsequent operations will be traversing down the updated directory path. For normal move operations, we can bound the downward phase by the length of the upward phase. Hence, we have
        \begin{equation}\label{eqn: cost alg strong}
            C(S) \leq h'D' + 2 \sum_{i=0}^{h'}|S_i|c\sigma r_i((I+f)+\sigma).
        \end{equation}
       
        From Equations \ref{eqn:lower_cost} and \ref{eqn: cost alg strong}, and since $r_i / r_{i-1} \leq \rho$, we get the competitive ratio for the move operations in $S$. For a strong sparse partition we have 
        \begin{align*}\label{eqn: competitive ratio strong}
            \frac{C(S)}{C^*(S)} \leq \frac{h'(h'+1) D'}{\sum_{i=0}^{h'}|S_i|r_{i-1}} +\frac{2(h'+1)\sum_{i=1}^{h'}|S_i|c\sigma r_i((I+f)\sigma)}{\sum_{i=0}^{h'}|S_i|r_{i-1}}
            = \mathcal{O}(h'\sigma\rho((I+f) + \sigma)),
        \end{align*}
        where we assume the second term to be the dominating one, which holds for a sufficiently large set of move operations $S$ (namely, $|S| = \Omega( h'^2 D')$).
    \end{proof}

    \moveweak*
    \begin{proof}
        The proof is identical to Theorem \ref{lem: cost move with failures strong}, except $|P_{k}(u_{i_j})|\leq fI$ in a weak sparse partition.\qedhere
    \end{proof}


\section{Analysis of Algorithm Without Faults}\label{sec: Analysis without faults}
    
    \subsection{Lookup}
        \begin{lemma}\label{lem: cost lookup without failure}
            A lookup operation finds the token with cost that is a $O(\sigma^2 \rho I)$ factor from optimal.
        \end{lemma}
        
        \begin{proof}
             Suppose node $u$ issues a lookup request and the directory path points toward node $v$ that holds the token, where $u \neq v$, and $\rho^{i-1} \leq d_G(u, v) \leq \rho^i$.
                  
            Let $w$ be the leader node of the directory path at level $i$. From Lemma \ref{lem: Length directory path}, the length of the directory path from $v$ up to  node $w$ is $c_1 \sigma \rho^i$, for some constant $c_1$. Hence, $d_G(v, w) \leq c_1 \sigma \rho^i$. Therefore, 
            $$d_G(u, w) 
                \leq d_G(u, v) + d_G(v, w) 
                \leq \rho^i + c_1 \sigma \rho^i
                \leq \rho^i \sigma + c_1 \sigma \rho^i 
                \leq (1 + c_1) \sigma \rho^i
                \leq c_2 \sigma \rho^i,$$
            for some constant $c_2 \geq 1 + c_1$.
           
            Let $s_w$ be the special parent of $w$, which is the leader of the cluster that includes $w$ at level  $i' = i +  \log_{\rho} (c' \sigma)$, where the constant $c'$ is such that $c' \geq c_2$. Since $r_{i'} = \min\{D, \rho^{i' }\}$ and $\rho^{i'} = \rho^{i + \log_{\rho} (c' \sigma)} = c' \sigma \rho^i $, the node $s_w$ is in the $r_{i'}$-neighborhood of node $u$. Therefore, when the lookup operation reaches level $i'$ it is guaranteed to discover the special parent $s_w$. The special parent provides a link to the directory node $w$.
           
            From Lemma \ref{lem: lookup path}, for some constant $c_3$, the cost of the upward part of the lookup operation from $u$ until reaching level $i'$ is 
            \begin{align*}
                \text{cost up to level $i'$}  \leq \sum_{j=0}^{i'} c_3 r_j \sigma I
                \leq \sum_{j=0}^{i'}  c_3 \rho^j \sigma I
                = \frac{c_3 (\rho^{i'+1}-1) \sigma I}{(\rho -1)}
                = \frac{c_3 (c' \sigma \rho^{i+1} -1) \sigma I}{(\rho -1)}
                = \mathcal{O}(\sigma^2 \rho^i I).
            \end{align*}
            The downward traversal cost of the lookup is proportional to the distance between $s_w$ and $w$, and the length of the directory path from $w$ to $v$ (which is at most $c_1 \sigma \rho^i$, as discussed above). Since $w$ and $s_w$ are both in the same level $i'$ cluster, the distance between them is at most $d_G(w, s_w) \leq \rho^{i'} = c' \sigma \rho^i$. Hence the downward lookup cost is at most $c_1 \sigma \rho^i + c' \sigma \rho^i = \mathcal{O}(\sigma \rho^i) $.  Combining the upward and downward cost, the overall lookup cost is  $\mathcal{O}(\sigma^2 \rho^i I)$.
    
            The optimal cost of finding the token is $d_G(u,v) > \rho^{i-1}$. Therefore, the lookup operation cost is within a factor of $\mathcal{O}(\sigma^2 \rho I)$ from the optimal cost.
        \end{proof}
        
    \subsection{Move}
        Consider a sequence of move requests $S = s_1, \ldots, s_q$, that execute in a sequential manner, so that $s_i$ starts only after $s_{i-1}$ completes, where $i > 0$. Let $s_0$ be a publish operation.
        \begin{lemma}\label{lem: cost move without failure}
            The total cost of the move operations in $S$ is a $\mathcal{O}( h \rho \sigma (\sigma + I))$ factor from optimal.
        \end{lemma}

        \begin{proof}
            Let $S_i = s_{i_1}, s_{i_2}, \ldots, s_{i_z}$, $i \geq 0$, be the sub-sequence of the move operations in $S$ that reach up to level $i \geq 0$ in their upward phase; namely, these operations modify a directory link at level $i$ to point to level $i-1$. Suppose that $s_{i_0} = s_0$.
                    
            Suppose that move operation $s_{i_j}$ is issued by node $u_{i_j}$. 
            Since operation $s_{i_j}$ reaches level $i$, it forms a directory path $\phi = p_{i_j}$ by linking the leader nodes in the clusters that $u_{i_j}$ participates to up to level $i-1$. Then $p_{i_j}$ links to a leader in layer $i$ that belongs to some cluster $X$ which is in the $r_i$-neighborhood of $u_{i_j}$.  
            
            We continue to show that $d(u_{i_{j-1}}, u_{i_j}) > r_{i-1}$, for $j > 0$. The reason is follows. Between $s_{i_{j-1}} $ and $s_{i_j}$ there is no other operation in $S$ between them that reaches level $i$, and hence when $p_{i_j}$ is formed the only previous directory path that reached level $i$ is $p_{i_{j-1}}$. Since $s_{i_j}$ reaches level $i$, it does not discover $p_{i_{j-1}}$ at level $i-1$. Before $s_{i_j}$, the only request that could have set the leader of the directory path at level $i-1$ is $s_{i_{j-1}}$, since otherwise there would have been another operation between $s_{i_{j-1}}$ and $s_{i_j}$ that reaches level $i$.  This implies that $u_{i_{j-1}}$ is not in the $r_{i-1}$-neighborhood of $u_{i,j}$. Consequently, $d(u_{i_{j-1}}, u_{i_j}) > r_{i-1}$.
            
            Let $C^*(S_i)$ denote the optimal cost of the operations in $S_i$. Since for any two consecutive operations in $S_i$ (including pair $s_{i_0}, s_{i_1}$) the distance of the source nodes is more than $r_{i-1}$, we have that $C^*(S_i) > |S_i| r_{i-1}$. Let $C^*$ be the optimal cost of all the operations in $S$. For the overall optimal cost we have 
            \begin{equation}
            \label{eqn:lower_cost}
                C^*(S)\geq \max_{0 \leq i\leq h } C^*(S_i)
                \geq \frac {\sum_{i=0}^{h} C^*(S_i)} {h+1}
                > \frac {\sum_{i=0}^{h} |S_i| r_{i-1}} {h+1}.
            \end{equation}
            
            Let $C(S_i)$ denote the cost of our directory scheme for serving the move requests in $S_i$ at level $i$ in the upward phase. For an operation $s_{i_j} \in S_i$ originating at $u_{i_j}$ we will count in $C(S_i)$ the cost of the operation $s_{i_j}$ at level $i$ only, as the cost of the operation at the other levels (below or above $i$) will be counted in $C(S_{k})$, $k \neq i$. Similar to Lemma \ref{lem: lookup path}, the move at level $i$ involves $c_1 r_i \sigma I$ cost checking the up to at most $I$ nearby cluster leaders, for a constant $c_1$ (we will also use additional constants $c_2, c_3, c_4$). We also have cost at most $c_2 r_i \sigma + r_i$ for linking the parent at level $i$ to level $i-1$. Moreover, for updating the special parent of $u_{i_j}$ which is at level $i + \log_\rho(c'\sigma)$, the involved cost is $c_3\sigma r_{i+\log_\rho (c' \sigma)}$. Note that $r_{i+\log_\rho (c' \sigma)} \leq r_{i}\rho^{\log_\rho (c' \sigma)} = r_{i} c' \sigma$, and hence the cost of updating the special parent is at most $c_3 c' r_i \sigma^2$. Adding all the above we have $c_1 r_i \sigma I + c_2 r_i \sigma + r_i + c_3 c' r_i \sigma^2 \leq c_4 r_i \sigma (\sigma + I)$, for a constant $c_4$. Therefore, $C(S_i) \leq |S_i| c_4 r_i \sigma(\sigma + I)$.
    
            Let $C(S)$ be the total cost of the move operations. Since a move operation will go through a level twice, once in the upward phase and once in the downward phase, and the downward phase cost does not exceed the upward phase cost, we consider the cost of $C(S_i)$ twice.
            Hence,
    		\begin{equation}\label{eqn:upper_cost}
                C(S) \leq 2 \cdot \sum_{i=0}^{h} C(S_i)
                \leq 2 \cdot \sum_{i=0}^{h} |S_i| c_4 r_i \sigma (\sigma + I).
            \end{equation}
            
            From Equations \ref{eqn:lower_cost} and \ref{eqn:upper_cost}, and since $r_i / r_{i-1} \leq \rho$, we get for the approximation of the total cost of the move operations in $S$:
            
            \begin{equation*}
                \frac{C(S)}{C^*(S)} < 2 (h+1) \frac{\sum_{i=0}^{h} |S_i| c_4 r_i \sigma (\sigma + I)}{\sum_{i=0}^{h} |S_i| r_{i-1}}
                \leq 2 c_4 (h+1) \rho \sigma (\sigma + I) = \mathcal{O}(h \rho \sigma (\sigma +I)).\qedhere
            \end{equation*}
        \end{proof}
    
  
\section{Cost Analysis of Fault Mechanism}\label{sec: cost fault mechanism}
    There is also a cost associated with our edge repair mechanism. In this section, we analyze the overhead cost of the repair encountered by a single edge failure $e=\{u, v\}$.

    \subsection{Cost of Updating the Shortest Path Tree}

        To update the shortest path trees we use the fully dynamic algorithm to maintain shortest path trees developed by King~\cite{king1999fully}. To update a single shortest path tree King's algorithm requires $\mathcal{O}(md)$ time, where $d$ denotes the maximum distance from the root of the tree to any other node, and $m$ denotes the total number of edges in the graph. Therefore, updating all shortest path trees can be done in $\mathcal{O}(mnD')$ time, where $D' = \diam(G\setminus\{e\}$) and $n$ denotes the total number of vertices in $G$.
    
        King's algorithm is a centralized algorithm that mimics Dijkstra. We can implement the same algorithm as a decentralized algorithm where the node root performs the computation of updating the shortest path tree. In this case, there will be an additional cost to inform the root node of the available edges. In particular, the root node would need to be informed about at most $\mathcal{O}(m)$ edges, and the maximum distance of such an edge to the root node is $\mathcal{O}(\diam(G'))$, where $G'$ denotes the current graph on which we are computing the shortest path tree. 
        
        After updating the shortest path trees, we also need to inform the nodes that are incident to edges that changed on the shortest path trees about the changes. That is if an edge $e_1=\{a, b\}$ was part of the shortest path tree before the edge failure but is no longer part of the shortest path tree after the edge failure, then $a$ and $b$ both need to be informed of the update. Similarly, if an edge $e_2=\{c, d\}$ was not part of the shortest path tree prior to the failure, but is part of the tree after the failure, then both $c$ and $d$ need to be informed of the failure.
        
        \begin{lemma}
           To inform all nodes about the updates of shortest path trees that affect their incident edges requires $\mathcal{O}(n^2)$ messages. Each message has size $\mathcal{O}(\log(n))$ and the total distance that a message needs to travel is $\diam(G')$, where $G'$ is the current graph on which we are computing the shortest path trees.
        \end{lemma}
        
        \begin{proof}
            For each tree that is changed, at most $\mathcal{O}(n)$ edges changed. Therefore, we need to send at most $\mathcal{O}(n)$ messages for each tree. Each message specifies information about a particular edge. Encoding the id of the edge requires $\mathcal{O}(\log(n))$ bits. The messages are sent along the newly computed shortest path trees. Therefore, the maximum distance a message needs to be sent is $\diam(G')$.
        \end{proof}

    \subsection{Cost of Reclustering}
        \begin{lemma}\label{lem: cost reclustering}
            To recluster a single level $i$ cluster in a strong sparse partition we need to send one message of size at most $\mathcal{O}(\log(n))$. This message will travel a distance of at most $\sigma r_i$. If the cluster is part of a weak sparse partition then the reclustering requires an additional message of size $\mathcal{O}(n\log(n))$ that needs to be sent a distance of at most $\sigma r_i$.
            Furthermore to inform all nodes in $X_2$ of the reclustering we need to send a total of $\mathcal{O}(n)$ messages of size $\mathcal{O}(\log(n))$. In a strong sparse partition, these messages need to travel a distance of at most $\sigma r_i$, and a distance of $2\sigma r_i$ in a weak sparse partition. 
        \end{lemma}
    
        \begin{proof}
            Let $X$ be a level $i$ cluster that needs to be reclustered due to a failure of edge $e=\{u, v\}$. Without loss of generality, assume $d(u, l(x)) < d(v, l(X))$. Then node $u$ sends a message to $l(X)$ that informs $l(X)$ and all nodes on the path connection $u$ and $l(X)$ of the change in the cluster. The distance between $u$ and $l(X)$ is at most $\diam(X)\leq \sigma r_i$. As every node $x$ on the path from $u$ to $l(X)$ knows $T_{x}(X)$, sending the id of $u$ along the path from $u$ to $l(X)$ on $T(X)$ suffices so that every node can update their knowledge of $T(X_1)$. Within $X_1$ no further messages need to be sent. 
            
            To update $X_2$, we differentiate between the cases $v\in X_2$ and $v\not\in X_2$. If $v\in X_2$, then $v$ will become the leader of $X_2$ and there is no need for $v$ to send a message to any other node. If $v\not\in X_2$, then it chooses some node $w$ in $X\cap T_v(X)$ to become the new leader node. In this case, $v$ will send a message $m$ to $w$ to inform it about its new leadership role. Message $m$ is sent along $T(X)$ and every node $x$ on the path from $v$ to $w$ will include information on $T_x(X)$. The reason for including this information is that our spanning tree for $T(X_2)$ is $T_v(X)$ rerooted at $w$. Because node $w$ is in $T_v(X)$, it must be that $v$ is on the path from $w$ to $l(X)$ on $T(X)$. Therefore, $d(v, w) \leq d(l(X), w)\leq \sigma r_i$. 
            
            We require that every node $y$ on $T(X_2)$ knows that it is part of $T(X_2)$ and that it knows its subtree i.e., $T_y(X_2)$. Because we are rerooting the tree at $w$, we need to inform it of the additional descendants it gets through the rerooting. If $n_2$ denotes the number of node in $X_2$, then $T(X_2)$ contains $n_2-1$ edges each of which can be decoded by its two endpoints, which requires $\mathcal{O}(\log(n))$ bits, where $n$ denotes the total number of nodes in $G$. Therefore in the worst case message $m$ has size $\mathcal{O}(n\log(n))$. 
            
            Once $l(X_2)$ is informed of the cluster splitting and whether or not it is part of the directory path, node $l(X_2)$ sends a broadcast message to all nodes in $X_2$ to inform them about the reclustering. This message includes the id of $l(X_2)$ which requires $\mathcal{O}(\log(n))$ bits. The total number of nodes in $X_2$ is $\mathcal{O}(n)$. In a strong sparse partition, the leader node of $X_2$ is node $v$. It must be that for all nodes in $X_2$, node $v$ is on the path to $l(X)$ on $T(X)$ prior to the failure of $e$. Therefore the distance from $v$ to any node $w$ in $X_2$ is smaller than $\sigma r_i$, the distance from $l(X)$ to $w$ on $T(X)$ before the failure of $e$. In a weak sparse partition, we know from Lemma \ref{lem: change diameter cluster} that the diameter of $X_2$ is at most $2\sigma r_i$.

            In case we are splitting the root level cluster, the root needs to inform all nodes in $X_1$ about the additional layers. This can also be achieved with $\mathcal{O}(n)$ messages of cost $\sigma \rho^h$. 
        \end{proof}
                
    \subsection{Cost of Updating the Directory Path and the Special Parent Information}
        Let $X$ be a level $i$ cluster that splits into $X_1$ and $X_2$ due to an edge failure of edge $e=\{u, v\}$.
 
        \begin{lemma}\label{lem: cost update directory path}
            To update the directory path the total distance that a message needs to travel is at most $\diam(G\setminus\{e\})$. The maximum size of a message is $\mathcal{O}(\log(n))$ and the total number of messages that we need to send is constant.
        \end{lemma}
        
        \begin{proof}
            When a level $i$ cluster $X$ splits into clusters $X_1$ and $X_2$ due to an edge failure, then $l(X)$ needs to inform $X_2$ if it is part of the directory path. In case $X_2$ is not part of the directory path node $l(X)$ simply sends a message to node $v$ informing it that cluster $X_2$ is not part of the directory path. In case node $v$ is not the leader node of cluster $X_2$ it forwards this information to $l(X_2)$. Because $l(X)$ and $v$ are no longer part of the same cluster after the edge failure, the only bound we have on the distance between $l(X)$ and $v$ is $d(l(X), v)\leq \diam(G\setminus\{e\})$. The distance between $v$ and $l(X_2)$ can be bounded by the diameter of $X_2$, that is $d(l(X_2), v) \leq 2\sigma r_i$. When $X_2$ is not part of the directory path, we do not need to update any special parent information.
        
            In the case that $X_2$ is part of the directory path, then $l(X)$ will first send a message to $\phi_{i+1}$ and $\phi_{i-1}$ to inform them of the update that is about to happen. This will block $\phi_{i+1}$ and $\phi_{i-1}$ from initializing a directory path update until the level $i$ update is completed. Due to the edge failure, the distance between consecutive leader nodes on a directory path that was formed before the edge failure can only be bounded by $\diam(G\setminus\{e\})$. In a second step, node $l(X)$ sends a message to $v$, this time informing $v$ that $l(X_2)$ is part of the directory path. In this case, the message also includes information about $\phi_{i+1}$ and $\phi_{i-1}$. This information can be encoded in $\mathcal{O}(\log(n))$ bits, where $n$ denotes the total number of nodes in the graph. In case $v$ is not $X_2$'s leader node, $v$ will again forward this message to $l(X_2)$. Node $l(X_2)$ will set pointers to $\phi_{i+1}$ and $\phi_{i-1}$, as well as send a message to these nodes, so they can also update their pointers. The distance between $l(X_2)$ and $\phi_{i+1}$, respectively between $l(X_2)$ and $\phi_{i-1}$ is at most $\diam(G\setminus\{e\})$. When $\phi_{i+1}$ and $\phi_{i-1}$ receive $l(X_2)$'s message, they update their pointer and send a message to $l(X)$ so that it removes its pointers.
        \end{proof}
        
        In case cluster $X_2$ becomes part of the directory path due to the edge failure, or if we add additional layers to the hierarchy we also need to update the special parent information. 

        \begin{lemma}\label{lem: cost update special parent}
            To update the special parent information we need to send two messages of constant size, and each message needs to travel a distance of at most $2\sigma r_{i'}$, where $i' = i + \log_\rho(c'\sigma)$.
        \end{lemma}
        
        \begin{proof}
            To update the special parent information nodes $l(X)$  needs to send a message to $l_{i'}(l(X))$ to inform it that $l(X)$ is no longer part of the directory path, and $l(X_2)$ needs to send a message to $l_{i'}(l(X_2))$ to inform it that $l(X_2)$ is part of the directory path. By Lemma \ref{lem: change diameter cluster}, the diameter of a cluster at level $i'$ is at most $2\sigma r_{i'}$, thus both of these messages need to traverse a distance of at most $2\sigma r_{i'}$, because $l(X)$ and $l_{i'}(l(X))$ are both in $C_{i'}(l(X))$ and $l(X_2)$ and $l_{i'}(l(X_2))$ are both in $C_{i'}(l(X_2))$.
        \end{proof}

    \subsection{Cost of Updating the Preprocessing Information}\label{subsec: cost update}
        Let $X$ be a level $i$ cluster that splits into clusters $X_1$ and $X_2$ due to an edge failure. Then all nodes in $X_2$ need to inform their $r_i$-neighborhood about the cluster change. 
        
        \begin{lemma}\label{lem: cost update preprocessing info}
            To inform the $r_i$ neighborhood about the cluster change of the nodes in $X_2$ we need to send at most $n^2$ messages. Each message has size $\mathcal{O}(\log(n))$ and needs to travel a distance of at most $r_i$.
        \end{lemma}
        
        \begin{proof}
            Each node in $X_2$ needs to send a message to every node in its $r_i$-neighborhood. In the worst case, $|X_2|=\mathcal{O}(n)$, and for every node in $X_2$, the $r_i$-neighborhood has size $\mathcal{O}(n)$. In this case, a total of $n^2$ messages need to be sent. Each message contains the id of $l(X_2)$ which requires $\mathcal{O}(\log(n))$ bits. As a node only needs a message to the nodes in its $r_i$-neighborhood, and because we send messages along shortest paths no message needs to travel further than $r_i$.
        \end{proof}


\section{Pseudocode of Basic Directory Algorithm}\label{sec: pesudocode}

    \begin{algorithm}[H]
        \small
        \caption{Directory Operations Issued by Node $v$} 
        Graph $G$ has partition hierarchy $\mathcal{P}$ with topmost level $h = \lceil \log_\rho{D} \rceil$, for constant $\rho > 1$\;
        Directory path $\phi = \phi_{-1}, \phi_{0}, \ldots, \phi_{h}$ points toward the current owner of token $t$\;  
        
        \BlankLine\BlankLine
        
        \tcp{\bf Publish Operation}
        {$\phi_{-1} \gets l_{-1}(v)$; }
        \For{level $i$ from $0$ to $h$}
        {$\phi_i \gets l_i(v)$; \tcp{$\phi_i$ is set to be the leader of $v$ at level $i$}
        Add bidirectional links between $\phi_i$ and $\phi_{i-1}$\;}
        \BlankLine\BlankLine
        
        \tcp{\bf Lookup Operation}
        
        $i \gets 0$; \tcp{start level of upward phase}
        \While{none of the leaders of clusters in $P_i{v}$ know about $\phi$}
        {i++; \tcp{upward phase to discover $\phi$}}
        
        If a special parent pointer toward $\phi_{i'}$ ($i' < i$) 
        was discovered at level $i$, then adjust $i \gets i'$\; 
        
        \tcp{downward phase toward token}
        \For{level $k \gets i$ down to $0$}
        {Follow the downward pointer of $\phi_k$\; }
        
        Return value of token $t$ from owner node $\phi_{-1}$\;
        
        \BlankLine\BlankLine
        
        \tcp{\bf Move Operation}
        
        $\phi_{-1} \gets v$; \tcp{start forming new $\phi$ toward $v$}
        $i \gets 0$; \tcp{upward phase to discover previous $\phi$} 
        \While{none of the leaders of clusters in $P_i(v)$ are $\phi_i$}
        {$\phi_i \gets l_i(v)$; \tcp{form new path $\phi$}
        Add bidirectional links between $\phi_i$ and $\phi_{i-1}$\;
        Inform special parent $l_{i'}$ at level $i' = i + \log_\rho (c'\sigma)$ about $\phi_i$\;
        $i++$\;}
        
        $old \gets$ level $i-1$ node pointed downwards by $\phi_i$\;
        
        Add bidirectional links between $\phi_i$ and $l_{i-1}(v)$; \tcp{adjust topmost node}
        
        Delete upward link of $old$ and information at special parent of $old$\;
        
        \tcp{downward phase to delete old directory path}
        \While{level of $old$ is not $-1$}
        {
        $w \gets$ node pointed by downward link of $old$\;
        Delete links between $w$ and $old$ and information at special parent of $w$\;
        $old \gets w$\;
        }
        Transfer token $t$ from $old$ to $v$; \tcp{$v$ is new owner}
    
    \end{algorithm}


\section{Conclusions}\label{sec: conclusion}

    We presented a fault-tolerant directory scheme based on sparse partitions that tolerates edge failures. We showed that the performance of the directory is linearly affected by the number of failures $f$. We showed how to adjust the clusters due to failures to transform the $\sigma$ and $I$ parameters, such that $\sigma$ simply doubles while $I$ is affected by either a $f$ factor (weak diameter clusters), or $f$ additive term (strong diameter clusters). 
    
    There are a few open questions that remain to be studied. One is to handle partitions of $G$ due to failures. The connected component that contains the token can still function and respond to operation requests. A related problem is examining the impact of node failures. If $G$ has bounded-degree $d$ a node failure corresponds to at most $d$ edge failures, then the techniques we developed could be adapted to analyze node failures.

    Another line of research related to preserving distances is building fault-tolerant sparse spanners. A sparse spanner of $G$ is a subgraph $H$ such that the pairwise distances on $G$ are stretched by a small factor on $H$. There exist fault-tolerant sparse spanners that maintain the stretch of the distances even after edge or node failures \cite{Bodwin2018Optimal,Parter2022Nearly}. Inspired by this, another direction is to design failure-oblivious sparse partitions with appropriate multiple pre-selected leaders in each cluster. Such leaders would be able to handle the failures without the need for cluster restructuring. We also believe that our work will help to analyze fault-tolerant sparse partitions in other settings than distributed directories.


@article{chung1984diameter,
  title={Diameter bounds for altered graphs},
  author={Chung, FRK and Garey, MR},
  journal={Journal of graph theory},
  volume={8},
  number={4},
  pages={511--534},
  year={1984},
  publisher={Wiley Online Library}
}

@inproceedings{king1999fully,
  title={Fully dynamic algorithms for maintaining all-pairs shortest paths and transitive closure in digraphs},
  author={King, Valerie},
  booktitle={40th Annual Symposium on Foundations of Computer Science (Cat. No. 99CB37039)},
  pages={81--89},
  year={1999},
  organization={IEEE}
}

@inproceedings{jia2005universal,
author = {Jia, Lujun and Lin, Guolong and Noubir, Guevara and Rajaraman, Rajmohan and Sundaram, Ravi},
title = {Universal Approximations for TSP, Steiner Tree, and Set Cover},
year = {2005},
isbn = {1581139608},
publisher = {Association for Computing Machinery},
address = {New York, NY, USA},
url = {https://doi.org/10.1145/1060590.1060649},
doi = {10.1145/1060590.1060649},
booktitle = {Proceedings of the Thirty-Seventh Annual ACM Symposium on Theory of Computing},
pages = {386–395},
numpages = {10},
location = {Baltimore, MD, USA},
series = {STOC '05}
}

@inproceedings{miller2013parallel,
author = {Miller, Gary L. and Peng, Richard and Xu, Shen Chen},
booktitle={Proceedings of the twenty-fifth annual ACM symposium on Parallelism in algorithms and architectures},
title = {Parallel Graph Decompositions Using Random Shifts},
year = {2013},
isbn = {9781450315722},
publisher = {Association for Computing Machinery},
address = {New York, NY, USA},
url = {https://doi.org/10.1145/2486159.2486180},
doi = {10.1145/2486159.2486180},
pages = {196–203},
numpages = {8},
location = {Montr\'{e}al, Qu\'{e}bec, Canada},
series = {SPAA '13}
}

@InProceedings{filtser2020scattering,
  author =	{Arnold Filtser},
  title =	{{Scattering and Sparse Partitions, and Their Applications}},
  booktitle =	{47th International Colloquium on Automata, Languages, and Programming (ICALP 2020)},
  pages =	{47:1--47:20},
  series =	{Leibniz International Proceedings in Informatics (LIPIcs)},
  ISBN =	{978-3-95977-138-2},
  ISSN =	{1868-8969},
  year =	{2020},
  volume =	{168},
  editor =	{Artur Czumaj and Anuj Dawar and Emanuela Merelli},
  publisher =	{Schloss Dagstuhl--Leibniz-Zentrum f{\"u}r Informatik},
  address =	{Dagstuhl, Germany},
  URL =		{https://drops.dagstuhl.de/opus/volltexte/2020/12454},
  URN =		{urn:nbn:de:0030-drops-124547},
  doi =		{10.4230/LIPIcs.ICALP.2020.47}
}

@article{Rai2022load,
  author    = {Shishir Rai and
               Gokarna Sharma and
               Costas Busch and
               Maurice Herlihy},
  title     = {Load balanced distributed directories},
  journal   = {Information and Computation},
  volume    = {285},
  number    = {A},
  year      = {2022},
  url       = {https://doi.org/10.1016/j.ic.2021.104700},
  doi       = {10.1016/j.ic.2021.104700},
  bibsource = {dblp computer science bibliography, https://dblp.org}
}

@article{Sharma2014distributed,
  author    = {Gokarna Sharma and
               Costas Busch},
  title     = {Distributed transactional memory for general networksking},
  journal   = {Distributed Computing},
  volume    = {27},
  number    = {5},
  pages     = {329--362},
  year      = {2014},
  url       = {https://doi.org/10.1007/s00446-014-0214-7},
  doi       = {10.1007/s00446-014-0214-7},
  bibsource = {dblp computer science bibliography, https://dblp.org}
}

@inproceedings{Demmer1998TheArrow,
  author    = {Michael J. Demmer and
               Maurice Herlihy},
  editor    = {Shay Kutten},
  title     = {The Arrow Distributed Directory Protocol},
  booktitle = {Distributed Computing, 12th International Symposium, {DISC} '98, Andros,
               Greece, September 24-26, 1998, Proceedings},
  series    = {Lecture Notes in Computer Science},
  volume    = {1499},
  pages     = {119--133},
  publisher = {Springer},
  year      = {1998},
  url       = {https://doi.org/10.1007/BFb0056478},
  doi       = {10.1007/BFb0056478},
  bibsource = {dblp computer science bibliography, https://dblp.org}
}

@article{Kerry1989TreeBased,
author = {Raymond, Kerry},
title = {A Tree-Based Algorithm for Distributed Mutual Exclusion},
year = {1989},
issue_date = {Feb. 1989},
publisher = {Association for Computing Machinery},
address = {New York, NY, USA},
volume = {7},
number = {1},
issn = {0734-2071},
url = {https://doi.org/10.1145/58564.59295},
doi = {10.1145/58564.59295},
journal = {ACM Transactions on Computer Systems},
pages = {61–77},
numpages = {17}
}

@article{Herlihy2006Dynamic,
  author    = {Maurice Herlihy and
               Fabian Kuhn and
               Srikanta Tirthapura and
               Roger Wattenhofer},
  title     = {Dynamic Analysis of the Arrow Distributed Protocol},
  journal   = {Theory Comput. Syst.},
  volume    = {39},
  number    = {6},
  pages     = {875--901},
  year      = {2006},
  url       = {https://doi.org/10.1007/s00224-006-1251-9},
  doi       = {10.1007/s00224-006-1251-9},
  bibsource = {dblp computer science bibliography, https://dblp.org}
}

@inproceedings{Peleg1999Variant,
author = {Peleg, David and Reshef, Eilon},
title = {A Variant of the Arrow Distributed Directory with Low Average Complexity},
year = {1999},
isbn = {3540662243},
publisher = {Springer-Verlag},
address = {Berlin, Heidelberg},
booktitle = {Proceedings of the 26th International Colloquium on Automata, Languages and Programming},
pages = {615–624},
numpages = {10},
series = {ICALP '99}
}

@article{Sharma2015Analysis,
  author    = {Gokarna Sharma and
               Costas Busch},
  title     = {An Analysis Framework for Distributed Hierarchical Directories},
  journal   = {Algorithmica},
  volume    = {71},
  number    = {2},
  pages     = {377--408},
  year      = {2015},
  url       = {https://doi.org/10.1007/s00453-013-9803-2},
  doi       = {10.1007/s00453-013-9803-2},
  bibsource = {dblp computer science bibliography, https://dblp.org}
}

@inproceedings{Ghodselahi2017Dynamic,
  author    = {Abdolhamid Ghodselahi and
               Fabian Kuhn},
  editor    = {Andr{\'{e}}a W. Richa},
  title     = {Dynamic Analysis of the Arrow Distributed Directory Protocol in General
               Networks},
  booktitle = {31st International Symposium on Distributed Computing, {DISC} 2017,
               October 16-20, 2017, Vienna, Austria},
  series    = {LIPIcs},
  volume    = {91},
  pages     = {22:1--22:16},
  publisher = {Schloss Dagstuhl - Leibniz-Zentrum f{\"{u}}r Informatik},
  year      = {2017},
  url       = {https://doi.org/10.4230/LIPIcs.DISC.2017.22},
  doi       = {10.4230/LIPIcs.DISC.2017.22},
  bibsource = {dblp computer science bibliography, https://dblp.org}
}

@inproceedings{Zhang2010Dynamic,
  author    = {Bo Zhang and
               Binoy Ravindran},
  title     = {Dynamic analysis of the relay cache-coherence protocol for distributed
               transactional memory},
  booktitle = {24th {IEEE} International Symposium on Parallel and Distributed Processing,
               {IPDPS} 2010, Atlanta, Georgia, USA, 19-23 April 2010 - Conference
               Proceedings},
  pages     = {1--11},
  publisher = {{IEEE}},
  year      = {2010},
  url       = {https://doi.org/10.1109/IPDPS.2010.5470393},
  doi       = {10.1109/IPDPS.2010.5470393},
  bibsource = {dblp computer science bibliography, https://dblp.org}
}

@article{Tirthapura2006Self,
  author    = {Srikanta Tirthapura and
               Maurice Herlihy},
  title     = {Self-Stabilizing Distributed Queuing},
  journal   = {{IEEE} Transactions on Parallel and Distributed Systems},
  volume    = {17},
  number    = {7},
  pages     = {646--655},
  year      = {2006},
  url       = {https://doi.org/10.1109/TPDS.2006.94},
  doi       = {10.1109/TPDS.2006.94},
  bibsource = {dblp computer science bibliography, https://dblp.org}
}

@article{Herlihy2007Distributed,
  author    = {Maurice Herlihy and
               Ye Sun},
  title     = {Distributed transactional memory for metric-space networks},
  journal   = {Distributed Comput.},
  volume    = {20},
  number    = {3},
  pages     = {195--208},
  year      = {2007},
  url       = {https://doi.org/10.1007/s00446-007-0037-x},
  doi       = {10.1007/s00446-007-0037-x},
  bibsource = {dblp computer science bibliography, https://dblp.org}
}

@article{Fakcharoenphol2004Tight,
  author    = {Jittat Fakcharoenphol and
               Satish Rao and
               Kunal Talwar},
  title     = {A tight bound on approximating arbitrary metrics by tree metrics},
  journal   = {J. Comput. Syst. Sci.},
  volume    = {69},
  number    = {3},
  pages     = {485--497},
  year      = {2004},
  url       = {https://doi.org/10.1016/j.jcss.2004.04.011},
  doi       = {10.1016/j.jcss.2004.04.011},
  bibsource = {dblp computer science bibliography, https://dblp.org}
}

@inproceedings{Khanchandari2020Arvy,
  author    = {Pankaj Khanchandani and
               Roger Wattenhofer},
  editor    = {Christian Scheideler and
               Petra Berenbrink},
  title     = {The Arvy Distributed Directory Protocol},
  booktitle = {The 31st {ACM} on Symposium on Parallelism in Algorithms and Architectures,
               {SPAA} 2019, Phoenix, AZ, USA, June 22-24, 2019},
  pages     = {225--235},
  publisher = {{ACM}},
  year      = {2019},
  url       = {https://doi.org/10.1145/3323165.3323181},
  doi       = {10.1145/3323165.3323181},
  bibsource = {dblp computer science bibliography, https://dblp.org}
}

@article{Kai1989Memory,
author = {Li, Kai and Hudak, Paul},
title = {Memory Coherence in Shared Virtual Memory Systems},
year = {1989},
issue_date = {Nov. 1989},
publisher = {Association for Computing Machinery},
address = {New York, NY, USA},
volume = {7},
number = {4},
issn = {0734-2071},
url = {https://doi.org/10.1145/75104.75105},
doi = {10.1145/75104.75105},
journal = {ACM Trans. Comput. Syst.},
pages = {321–359},
numpages = {39}
}

@inproceedings{Parter2022Nearly,
  author    = {Merav Parter},
  editor    = {Stefano Leonardi and
               Anupam Gupta},
  title     = {Nearly optimal vertex fault-tolerant spanners in optimal time: sequential,
               distributed, and parallel},
  booktitle = {{STOC} '22: 54th Annual {ACM} {SIGACT} Symposium on Theory of Computing,
               Rome, Italy, June 20 - 24, 2022},
  pages     = {1080--1092},
  publisher = {{ACM}},
  year      = {2022},
  url       = {https://doi.org/10.1145/3519935.3520047},
  doi       = {10.1145/3519935.3520047},
  bibsource = {dblp computer science bibliography, https://dblp.org}
}

@inproceedings{Bodwin2018Optimal,
  author    = {Greg Bodwin and
               Michael Dinitz and
               Merav Parter and
               Virginia Vassilevska Williams},
  editor    = {Artur Czumaj},
  title     = {Optimal Vertex Fault Tolerant Spanners (for fixed stretch)},
  booktitle = {Proceedings of the Twenty-Ninth Annual {ACM-SIAM} Symposium on Discrete
               Algorithms, {SODA} 2018, New Orleans, LA, USA, January 7-10, 2018},
  pages     = {1884--1900},
  publisher = {{SIAM}},
  year      = {2018},
  url       = {https://doi.org/10.1137/1.9781611975031.123},
  doi       = {10.1137/1.9781611975031.123},
  bibsource = {dblp computer science bibliography, https://dblp.org}
}

@inproceedings{Kuhn2004Dynamic,
  author    = {Fabian Kuhn and
               Roger Wattenhofer},
  editor    = {Phillip B. Gibbons and
               Micah Adler},
  title     = {Dynamic analysis of the arrow distributed protocol},
  booktitle = {{SPAA} 2004: Proceedings of the Sixteenth Annual {ACM} Symposium on
               Parallelism in Algorithms and Architectures, June 27-30, 2004, Barcelona,
               Spain},
  pages     = {294--301},
  publisher = {{ACM}},
  year      = {2004},
  url       = {https://doi.org/10.1145/1007912.1007962},
  doi       = {10.1145/1007912.1007962},
  bibsource = {dblp computer science bibliography, https://dblp.org}
}

@inproceedings{Herlihy2001Competitive,
author = {Herlihy, Maurice and Tirthapura, Srikanta and Wattenhofer, Rogert},
title = {Competitive Concurrent Distributed Queuing},
year = {2001},
ISBN = {1581133839},
publisher = {Association for Computing Machinery},
address = {New York, NY, USA},
URL = {https://doi.org/10.1145/383962.384001},
DOI = {10.1145/383962.384001},
booktitle = {Proceedings of the Twentieth Annual ACM Symposium on Principles of Distributed Computing},
pages = {127–133},
numpages = {7},
location = {Newport, Rhode Island, USA},
series = {PODC '01}
}

@article{Sharma2015Near,
  author    = {Gokarna Sharma and
               Hari Krishnan and
               Costas Busch and
               Steven R. Brandt},
  title     = {Near-Optimal Location Tracking Using Sensor Networks},
  journal   = {International Journal of Networking and Computing},
  volume    = {5},
  number    = {1},
  pages     = {122--158},
  year      = {2015},
  url       = {http://www.ijnc.org/index.php/ijnc/article/view/100},
  bibsource = {dblp computer science bibliography, https://dblp.org}
}

@inproceedings{Baruch1991Concurrent,
author = {Awerbuch, Baruch and Peleg, David},
title = {Concurrent Online Tracking of Mobile Users},
year = {1991},
isbn = {0897914449},
publisher = {Association for Computing Machinery},
address = {New York, NY, USA},
url = {https://doi.org/10.1145/115992.116013},
doi = {10.1145/115992.116013},
booktitle = {Proceedings of the Conference on Communications Architecture \& Protocols},
pages = {221–233},
numpages = {13},
location = {Zurich, Switzerland},
series = {SIGCOMM '91}
}
\end{document}